\documentclass[english,final]{IEEEtran}
\usepackage[T1]{fontenc}
\usepackage{mathrsfs}
\usepackage{bm}
\usepackage{amsmath}
\usepackage{amsthm}
\usepackage{amssymb}
\usepackage{stmaryrd}
\usepackage{graphicx}
 \PassOptionsToPackage{normalem}{ulem}
 \usepackage{ulem}

\makeatletter

\theoremstyle{plain}
\newtheorem{theorem}{\protect\theoremname}
  \theoremstyle{plain}
  
  \theoremstyle{plain}
  
  \theoremstyle{plain}
   \newtheorem{lemma}{\protect\lemmaname}
  \theoremstyle{remark}
  
   \newtheorem{assumption}{\protect\assumptionname}
\theoremstyle{assumption}

\theoremstyle{algorithm}  
  
\usepackage{amsmath}
\usepackage{amssymb}
\usepackage{graphicx,psfrag,cite,subfigure}
\usepackage[table]{xcolor}

\usepackage[acronym]{glossaries}
\newcommand{\newac}{\newacronym}

\makeglossaries

\newac{speb}{SPEB}{square position error bound}
\newac[plural=EFIMs,firstplural=Fisher information matrices (EFIMs)]{efim}{EFIM}{Fisher information matrix}
\newac{ne}{NE}{Nash equilibrium}
\newac{mse}{MSE}{mean squared error}
\newac{toa}{TOA}{time-of-arrival}
\newac{snr}{SNR}{signal-to-noise ratio}
\newac{lan}{LAN}{local area network}
\newac{psd}{PSD}{positive semidefinite}
\newac{pd}{PD}{positive definite}
\newac{wrt}{w.r.t.}{with respect to}
\newac{lhs}{L.H.S.}{left hand side}
\newac{wp1}{w.p.1}{with probability 1}
\newac{kkt}{KKT}{Karush-Kuhn-Tucker}
\newac{wlog}{w.l.o.g.}{without loss of generality}
\newac{mle}{MLE}{maximum likelihood estimation}
\newac{gps}{GPS}{global positioning system}
\newac{rssi}{RSSI}{received signal strength}
\newac{mimo}{MIMO}{multiple-input multiple-output}
\newac{csi}{CSI}{channel state information}
\newac{fdd}{FDD}{frequency division duplexing}
\newac{ms}{MS}{mobile station}
\newac{bs}{BS}{base station}
\newac{d2d}{D2D}{device-to-device}
\newac{slnr}{SLNR}{signal-to-interference-leakage-and-noise-ratio}
\newac{ula}{ULA}{uniform linear antenna array}
\newac{pas}{PAS}{power angular spectrum}
\newac{mmse}{MMSE}{minimum mean square error}
\newac{zf}{ZF}{zero-forcing}
\newac{rzf}{RZF}{regularized zero-forcing}
\newac{as}{AS}{angular spread}
\newac{aod}{AOD}{angle of departure}
\newac{iid}{i.i.d.}{independent and identically distributed} 
\newac{sinr}{SINR}{signal-to-interference-and-noise ratio}
\newac{tdd}{TDD}{time-division duplex}
\newac{rvq}{RVQ}{random vector quantization}
\newac{rhs}{R.H.S.}{right hand side}
\newac{mrc}{MRC}{maximum ratio combining}
\newac{cdf}{CDF}{cumulative distribution function}
\newac{a.s.}{a.s.}{almost surely}
\newac{los}{LOS}{line-of-sight}
\newac{jsdm}{JSDM}{joint spatial division and multiplexing}
\newac{map}{MAP}{maximum a posteriori}
\newac{klt}{KLT}{Karhunen-Lo\`eve Transform}
\newac{lbe}{LBE}{link bargaining equilibrium}
\newac{se}{SE}{Stackelberg equilibrium}
\newac{uav}{UAV}{unmanned aerial vehicle}
\newac{nlos}{NLOS}{non-line-of-sight}
\newac{pdf}{PDF}{probability density function}
\newac{em}{EM}{expectation-maximization}
\newac{knn}{KNN}{$k$-nearest neighbor}
\newac{svd}{SVD}{singular value decomposition}
\newac{nmf}{NMF}{non-negative matrix factorization}
\newac{umf}{UMF}{unimodal-constrained matrix factorization}
\newac{rmse}{RMSE}{rooted mean squared error}

\setkeys{Gin}{width=1.0\columnwidth}

\makeatother

\usepackage{babel}
  \providecommand{\definitionname}{Definition}
  \providecommand{\lemmaname}{Lemma}
  \providecommand{\propositionname}{Proposition}
  \providecommand{\remarkname}{Remark}
\providecommand{\theoremname}{Theorem}
\providecommand{\conjecturename}{Conjecture}
\providecommand{\assumptionname}{Assumption}
\providecommand{\algorithmname}{Algorithm}

\begin{document}

 \title{Security against false data injection attack in cyber-physical systems
 \thanks{The authors are with the Department of Electrical Engineering, University of 
 Southern California. Email: \{achattop,ubli\}@usc.edu}
 \thanks{This work was supported by the following funding sources: ONR N00014-15-1-2550, 
NSF CNS-1213128,
NSF CCF-1718560,
NSF CCF-1410009,
NSF CPS-1446901,
AFOSR FA9550-12-1-0215
}
\thanks{This paper is an extension of our previous conference paper \cite{chattopadhyay2018attack}.}
}

\author{
Arpan~Chattopadhyay \& Urbashi~Mitra  \vspace*{-0.4in}
}

\maketitle
%
%



\ifdefined\SINGLECOLUMN
	\setkeys{Gin}{width=0.5\columnwidth}
	\newcommand{\figfontsize}{\footnotesize} 
\else
	\setkeys{Gin}{width=1.0\columnwidth}
	\newcommand{\figfontsize}{\normalsize} 
\fi

\begin{abstract} 
In this paper,  secure, remote estimation  of a linear  Gaussian process via observations at multiple sensors is considered.  Such a framework is relevant to many cyberphysical systems and internet-of-things applications.  Sensors make sequential measurements that are shared with a fusion center; the fusion center applies a  certain  filtering algorithm to make its estimates.   The challenge is the presence of a few {\em unknown} malicious sensors which can inject anomalous observations to skew the estimates at the fusion center.  The set of malicious sensors may be time-varying.  The problems of malicious sensor detection and secure estimation are considered. First,  an algorithm for secure estimation is proposed. The proposed estimation scheme uses a novel filtering and learning algorithm,  where an optimal filter is learnt over time by using the sensor observations in order to filter out malicious sensor observations while retaining other sensor measurements. Next, a novel detector to detect injection attacks on an unknown sensor subset is developed.      Numerical results demonstrate up to $3$~dB gain in the mean squared error and up to $75 \%$ higher   attack detection probability under a small false alarm rate constraint, against a competing algorithm that requires additional side information.
\end{abstract}

\vspace{-3mm}

\section{Introduction}\label{section:introduction}
Cyber-physical systems (CPS)   combine the cyber world and the physical world via seamless integration of sensing, control, communication and computation. CPS has widespread applications such as networked  monitoring and control of industrial processes, intelligent transportation systems,  smart grid, and environmental monitoring. Most of these applications critically depend on reliable remote  estimation of a physical process via multiple sensors over a wireless network. Hence,  any malicious attack on sensors can have a catastrophic impact. We focus on {\em false data injection} (FDI) attacks which can be characterized as integrity or deception attacks where the attacker modifies the information sent to the fusion center \cite{mo2009secure, mo2014detecting}. This is in contrast to a {\em denial-of-service} attack where the attacker attempts to block resources for the system (e.g., wireless jamming attack to block bandwidth usage \cite{guan2017distributed}). FDI attacks modify the information either by breaking the cryptography of the data packets or by physical manipulation of the sensors (e.g., putting a heater near a temperature sensor).

The problem of FDI attack and its detection has received recent attention. In  \cite{chen2017optimal},    conditions for    undetectable FDI  attack are developed, and   the minimum number of sensors to be attacked to ensure undetectability is computed. In \cite{guo2017optimal}, a linear deception attack scheme that can fool the popular $\chi^2$ detector is provided. Later,  a new detection algorithm against such linear deception attacks is designed in \cite{li2017detection}, where observations are available from a few known {\em safe} sensor nodes. 
Efficient attack detection and secure estimation schemes for linear Gaussian systems under cyber attack on a static, unknown sensor subset have been developed in \cite{mishra2017secure}, but this detector is not designed to tackle the linear deception attack of \cite{guo2017optimal}.   The   optimal attack strategy to steer the control of CPS to a target value is provided in \cite{chen2016cyber}, while ensuring a constraint on the attack detection probability. Centralized and decentralized attack detection schemes for {\em noiseless} systems have been developed in  \cite{pasqualetti2013attack}. Coding of sensor output for efficient attack detection using $\chi^2$ detector is proposed in \cite{miao2017coding}. Attack-resilient state estimation of a dynamical system with only  {\em bounded} noise has been discussed in \cite{pajic2017attack}.  Sparsity models to characterize the switching location attack in a {\em noiseless} linear system and    state recovery constraints for various attack modes have been described in \cite{liu2017dynamic}.   Attack detection, secure estimation and control in the presence of FDI attack for power  systems are addressed in \cite{manandhar2014detection, liang2017review, hu2017secure}.

In contrast to the prior literature, our current paper addresses the problem of attack detection and secure remote estimation of a linear system with {\em unbounded} Gaussian noise, in the presence of an FDI attack (which could be the FDI attack of \cite{guo2017optimal}), when no safe sensor subset is available. In this paper, we make the following contributions. (i) We   develop a learning algorithm that learns a Kalman-like filter over time for secure estimation in presence of FDI attacks.  The filter gain matrix is updated iteratively over time via simultaneous perturbation stochastic approximation (SPSA), in order to minimize a combination of the estimation error in the absence of attack and the anomaly in estimates returned by various sensor subsets.   The convergence of the learning algorithm is proved by exploiting the properties of the attack scheme, the filtering scheme and SPSA. To our knowledge, this  is a novel contribution to the adaptive filtering literature as well. This algorithm later motivates another low-complexity heuristic which offers up to $3$~dB improvement  in mean squared error (MSE) against another competing algorithm from \cite{li2017detection} that additionally requires a subset of {\em safe} sensors. The algorithms are   extended to handle random packet loss between the sensors and the fusion center.
(ii)  We propose an algorithm for FDI attack detection, that can offer up to $75 \%$ improvement in attack detection probability under a small false alarm constraint, against the detection scheme of \cite{li2017detection} which  additionally requires a subset of {\em safe} sensors.  Our detection  algorithm detects an attack via  anomaly detection between estimates made by various  sensor subsets. We also provide a learning scheme for optimization of our detector subject to a constraint on  false alarm.

The rest of the paper is organized as follows. Preliminaries are discussed in Section~\ref{section:background}. The secure estimation algorithm to combat the FDI attack is  described in Section~\ref{section:secure-estimation}. The attack detection scheme is described in Section~\ref{section:attack-detection}.  Numerical results are provided in Section~\ref{section:numerical-work}, followed by the conclusions in Section~\ref{section:conclusion}. All proofs are provided in the appendices.

\section{Background}\label{section:background}
Throughout this paper, bold capital letters, bold small letters and capital letters with caligraphic font will denote matrices, vectors and sets respectively.
\subsection{Sensing and remote estimation model}\label{subsection:sensing-model}
We consider  a set of smart sensors $\mathcal{N}:=\{1,2,\cdots,N\}$, which are sensing a discrete-time stochastic process $\{\mathbf{x}(t)\}_{t \geq 0}$. The sensors  send their observation directly to a fusion center via {\em error-free} wireless links so that the fusion center can estimate $\mathbf{\hat{x}}(t)$ at each time $t$. 
The physical process $\{\mathbf{x}(t)\}_{t \geq 0}$ (where $\mathbf{x}(t) \in \mathbb{R}^q$)  is a linear Gaussian process that evolves according to the following   equation:
\begin{equation}
\mathbf{x}(t+1)=\mathbf{A} \mathbf{x}(t)+\mathbf{w}(t),  \label{eqn:process-equation}
\end{equation}
where $\mathbf{w}(t)$ is a zero-mean Gaussian noise vector with covariance matrix $\mathbf{Q}$, and is i.i.d. across $t$. The scalar or vector observation made by sensor~$i$  is given by the following observation equation if sensor~$i$ is used in sensing:
\begin{equation}
\mathbf{y}_i(t)=\mathbf{C}_i \mathbf{x}(t)+\mathbf{v}_i(t),  \label{eqn:observation-equation}
\end{equation}
where $\mathbf{C}_i$ is a matrix of appropriate dimension and $\mathbf{v}_i(t)$ is a Gaussian observation noise with covariance matrix $\mathbf{R}_i$. Observation noise $\mathbf{v}_i(t)$ is assumed to be independent across sensors and i.i.d. across time. The pair $(\mathbf{A},\mathbf{Q}^{\frac{1}{2}})$ is assumed to be stabilizable and the pair $(\mathbf{A},\mathbf{C}_i)$ is assumed to be detectable for all $i \in \mathcal{N}$.

The goal of the fusion center   is to minimize the time-average expected mean squared error (MSE) in estimation:
\begin{equation}
 \limsup_{T \rightarrow \infty} \frac{1}{T} \sum_{t=0}^T \mathbb{E}||\mathbf{x}(t)-\mathbf{\hat{x}}(t)||^2 . \label{eqn:time-average-MSE}  
\end{equation}

If all sensors send their observations to the fusion center in real time, then the system is equivalent to a single sensor and a remote estimator with real-time communication. The sensing and observation models can be rewritten as:
\begin{eqnarray}
\mathbf{x}(t+1)&=& \mathbf{A} \mathbf{x}(t)+\mathbf{w}(t)\nonumber\\
\mathbf{y}(t)&=& \mathbf{C} \mathbf{x}(t)+\mathbf{v}(t), \label{eqn:single-sensor-model}
\end{eqnarray}
where $\mathbf{v}(t) \sim N(\mathbf{0},\mathbf{R})$ is the   observation noise and $\mathbf{y}(t) \in \mathbb{R}^{m \times 1}$ (also called $\mathbf{y}_t$) is the complete observation vector. The minimum mean-squared error (MMSE) estimator in this case is a linear filter called Kalman filter (see \cite{anderson1979optimal}):
\begin{eqnarray}
\mathbf{\hat{x}}_{t+1|t}&=& \mathbf{A} \mathbf{\hat{x}}_t \nonumber\\
\mathbf{P}_{t+1|t} &=& \mathbf{A} \mathbf{P}_t \mathbf{A}'+\mathbf{Q} \nonumber\\
\mathbf{K}_{t+1} &=& \mathbf{P}_{t+1 | t} \mathbf{C}' (\mathbf{C} \mathbf{P}_{t+1|t} \mathbf{C}' +\mathbf{R})^{-1} \nonumber\\
\mathbf{\hat{x}}_{t+1} &=& \mathbf{\hat{x}}_{t+1|t} +\mathbf{K}_{t+1} (\mathbf{y}(t+1)-\mathbf{C} \mathbf{\hat{x}}_{t+1|t} )  \nonumber\\
\mathbf{P}_{t+1} &=& (\mathbf{I}-\mathbf{K}_{t+1} \mathbf{C}) \mathbf{P}_{t+1|t}, \label{eqn:kalman-filter-for-single-sensor}
\end{eqnarray}
where $\mathbf{\hat{x}}_{t+1}$ is the MMSE estimate and $\mathbf{P}_{t+1}$ is the the  error covariance matrix for the estimate $\mathbf{\hat{x}}_{t+1}$, provided that the iteration starts from $\mathbf{\hat{x}}_0 \sim N(\mathbf{0},\mathbf{P}_0)$. It has been shown in \cite{anderson1979optimal} that   $\lim_{t \rightarrow \infty} \mathbf{P}_{t+1|t} =\mathbf{\underbar{P}}$ exists and is the unique fixed point to the 
$\mathbf{P}_{t+1|t}$ iteration called the {\em Riccati equation}. Another quantity of interest is the innovation sequence $\mathbf{z}_t:=\mathbf{y}(t)-\mathbf{C} \mathbf{\hat{x}}_{t|t-1}$;  it  was proved in \cite{anderson1979optimal} that $\{\mathbf{z}_t\}_{t \geq 1}$ is a zero-mean Gaussian sequence which is pairwise independent across time and whose covariance matrix in the steady state is $\mathbf{\Sigma}_{\mathbf{z}}:=(\mathbf{C} \mathbf{\underbar{P}}\mathbf{C}'+\mathbf{R})$.

\subsection{False data injection (FDI) attack}\label{subsection:FDI-attack}
Any {\em unknown} subset $\mathcal{A}(t) \subset \mathcal{N}$ of sensors can be under attack at time $t$. Any sensor~$i \in \mathcal{A}(t)$ sends an  observation according to the following attack equation: 
\begin{equation}
\mathbf{y}_i(t+1)=\mathbf{C}_i \mathbf{x}(t)+\mathbf{e}_i (t)+\mathbf{v}_i(t),  \label{eqn:attack-equation}
\end{equation}
where $\mathbf{e}_i(t)$ is an error term injected by the attacker. The goal of the attacker is to insert the false data sequence $\{\mathbf{e}_i(t): i \in \mathcal{A}(t)\}_{t \geq 0}$   so as to maximize the MSE given by  \eqref{eqn:time-average-MSE}. If $\mathcal{A}(t)=\mathcal{A}$ for all $t$, then the attack is called a {\em static attack}, otherwise the attack is called a {\em switching location attack}. We assume that, at most $n_0$ sensors can be under attack at a   time.


{\em The $\chi^2$ detector:}
Since $\mathbf{z}_t \sim N(\mathbf{0}, \mathbf{\Sigma}_{\mathbf{z}})$ under steady state when there is no attack, a natural technique (see \cite{guo2017optimal}, \cite{li2017detection}) to detect any FDI attack is to detect any anomaly in $\{\mathbf{z}_t\}_{t \geq 0}$. This is done by observing the innovation sequence over a pre-specified window of $J$ time-slots, and declaring   an attack at time $\tau$   if  
$\sum_{t=\tau-J+1}^{\tau} \mathbf{z}_t' \mathbf{\Sigma}_{\mathbf{z}}^{-1} \mathbf{z}_t \geq \eta$,  
where $\eta$ is a pre-specified threshould used to tune the false alarm probability.

\begin{figure}[t!]
\begin{centering}
\begin{center}
\includegraphics[height=3cm, width=6cm]{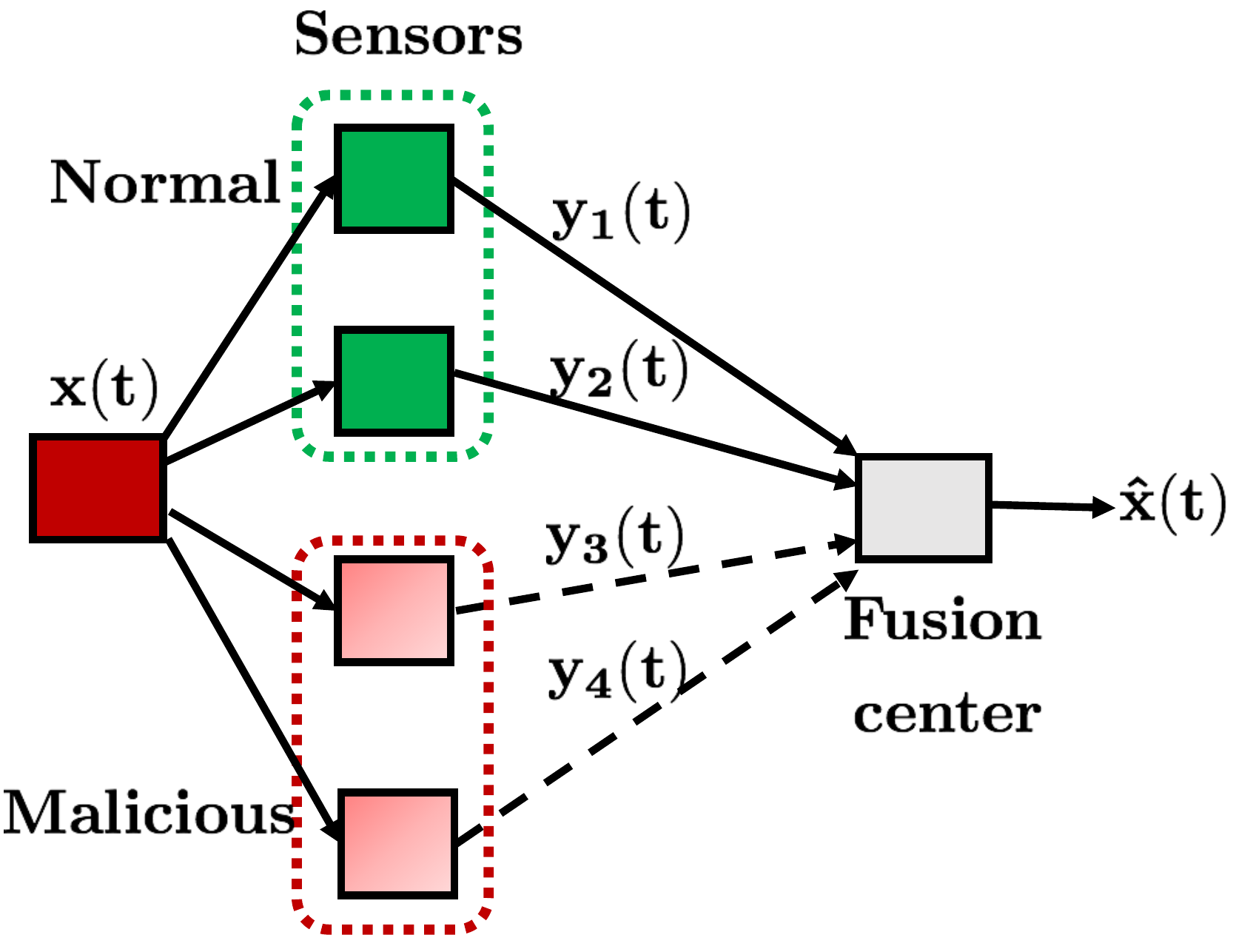}
\end{center}
\end{centering}
\caption{False data injection attack in remote estimation.}
\label{fig:FDI-attack-image}
\end{figure}

In \cite{guo2017optimal}, a linear injection attack to fool the $\chi^2$ detector is 
constructed; at time $t$, the   malicious sensor(s)    modifies the innovation vector as $\mathbf{\tilde{z}}_t=\mathbf{T} \mathbf{z}_t+\mathbf{b}_t$, where $\mathbf{T}$ is a square matrix and $\mathbf{b}_t \sim N (\mathbf{0},\mathbf{\Sigma}_{\mathbf{b}})$ is  i.i.d. Gaussian.  It was shown in \cite{guo2017optimal} that $\mathbf{\tilde{z}}_t \sim N(\mathbf{0},\mathbf{\Sigma}_{\mathbf{\tilde{z}}})$ where $\mathbf{\Sigma}_{\mathbf{\tilde{z}}}=\mathbf{T} \mathbf{\Sigma}_{\mathbf{z}} \mathbf{T}'+\mathbf{\Sigma}_{\mathbf{b}}$. Hence, if we ensure  $\mathbf{\Sigma}_{\mathbf{\tilde{z}}}=\mathbf{\Sigma}_{\mathbf{z}}$, then   $\{\mathbf{\tilde{z}}_t\}_{t \geq 1}$ will have the same distribution $N(\mathbf{0}, \mathbf{\Sigma}_{\mathbf{z}})$ as $\mathbf{z}(t)$, and hence the detection probability of the $\chi^2$ detector will remain unaffected even under the attack.   The estimation error   is maximized when the attacker just inverts the innovation sequence (i.e., when $\mathbf{T}=-\mathbf{I}$ and $\mathbf{b}_t=\mathbf{0}$).    However, the authors of \cite{li2017detection} proposed another efficient  scheme to detect such  attack. The detection algorithm in \cite{li2017detection} assumed the presence of a few {\em safe} sensors; an attack is detected by exploiting any anomaly between the observations made by the safe sensors and other sensors. The assumption of the existence of a set of safe sensors is restrictive, and the design of  efficient attack detection and secure estimation schemes in the  absence of such safe sensors is the topic of our current paper.

{\em Stationary attacks:} In this paper, we consider a special class of attack schemes called {\em stationary attacks}, where, at time $t$, the injected error $\{\mathbf{e}_i(t): i \in \mathcal{A}(t)\}$ is independently chosen from a distribution $p(\cdot|  \mathbf{y}_{\tau}: \tau < t)$. Note that, this attack class contains the class of linear attacks where the attacked sensor subset is either static, or varying with time according to a i.i.d. process or a time-homogeneous Markov chain.

\section{Secure estimation}\label{section:secure-estimation}
In this section, we will provide an algorithm to obtain a reliable estimate $\mathbf{\hat{x}}(t)$   in  the presence of FDI attacks, without explicitly detecting the malicious sensor subset. This algorithm is useful when it is not possible for the system administrator to take necessary measure even upon the detection of an attack (e.g., if a heater is deliberately kept by an attacker near a temperature sensor, it may not be always be  possible to physically remove the heater).

\subsection{Formulation as an on-line optimization problem}
Note that, any sensor observation is ignored if the corresponding entries in the Kalman gain matrix $\mathbf{K}_{t+1}$ in \eqref{eqn:kalman-filter-for-single-sensor} are set to $0$. Ideally, one should compensate for the errors introduced by the malicious sensors, and if done so, the anomalies in estimates from various sensor subsets will be small. However, since the estimation error depends collectively on  the process noise, the observation noise and the noise injected by the attacker, a reasonable technique for reliable estimation would be to dynamically learn an optimal   gain matrix that minimizes  the anomalies in estimates returned by various sensor subsets subject to a constraint on  the MSE in the absence of attack.

Now, let us assume that the attack is a stationary attack (a static attack is a special case). We restrict the discussion   to the class of linear estimators. The estimator we consider is similar to the Kalman filter in \eqref{eqn:kalman-filter-for-single-sensor}, except that the Kalman gain matrix $\mathbf{K}_{t+1}$ is learnt via a stochastic gradient descent scheme to solve the following constrainted optimization problem:

\footnotesize
\begin{eqnarray}
 &&\min_{\{\mathbf{K}_t\}_{t \geq 0}}\limsup_{T \uparrow  \infty}  \frac{1}{T} \sum_{t=0}^{T} \mathbb{E} [ \max_{\mathcal{B} \in 2^{\mathcal{N}}:|\mathcal{B}|=n_0}  ||\mathbf{\hat{x}}_{\mathcal{B}}(t)-\mathbf{\hat{x}}_{\mathcal{B}^c}(t)||^2 ] \nonumber\\
 &s.t. & \limsup_{T \uparrow  \infty} \frac{1}{T} \sum_{t=0}^{T}  \mathbb{E}
[ \mbox{Tr}(\mathbf{P}_t) ] \leq \bar{P}\label{eqn:constrained-problem-for-estimate-correction}
\end{eqnarray}
\normalsize

Here $\mathbf{P}_t$ represents the    estimate of the error covariance matrix  if a   gain matrix sequence $\mathbf{K}_0, \mathbf{K}_1,\mathbf{K}_2, \cdots$  is used for estimation, {\em when there is no attack}. Note that,   $\mathbf{P}_t$ can be calculated iteratively. 
The objective function $\mathbb{E} [ \max_{\mathcal{B} \in 2^{\mathcal{N}}:|\mathcal{B}|=n_0}  ||\mathbf{\hat{x}}_{\mathcal{B}}(t)-\mathbf{\hat{x}}_{\mathcal{B}^c}(t)||^2 ]$ captures the anomaly between the estimates $\mathbf{\hat{x}}_{\mathcal{B}}(t)$ and $\mathbf{\hat{x}}_{\mathcal{B}^c}(t)$ coming from two different sensor subsets $\mathcal{B}$ and $\mathcal{B}^c$, {\em when the restrictions $\mathbf{K}_{t,  \mathcal{B}}$ and $\mathbf{K}_{t, \mathcal{B}^c}$ of $\{\mathbf{K}_t\}_{t \geq 0}$ to the subsets $\mathcal{B}$ and $\mathcal{B}^c$ are used as   gain matrices applied to the observations coming from these sensor subsets.} The constraint $\bar{P}$ can be chosen to be a certain multiple of the time-average MMSE of the system under no attack (which can be computed by running an optimal Kalman filter).

The above constrained problem can be relaxed by a Lagrange multiplier $\lambda \geq 0$ as follows:

\footnotesize
\begin{equation}
\min_{\{\mathbf{K}_t\}_{t \geq 0}} \limsup_{\tau \uparrow  \infty} \frac{1}{\tau} \sum_{t=0}^{\tau} \mathbb{E} \underbrace {   
[ \max_{\mathcal{B} \in 2^{\mathcal{N}}:|\mathcal{B}|=n_0} ||\mathbf{\hat{x}}_{\mathcal{B}}(t)-\mathbf{\hat{x}}_{\mathcal{B}^c}(t)||^2 + \lambda \mbox{Tr}(\mathbf{P}_t)]   }_{:=c(t)}\label{eqn:unconstrained-problem-for-estimate-correction}
\end{equation}
\normalsize

where  $c(t)$ is the single-stage cost at time $t$. 

The following  result tells us how to choose $\lambda$.
\begin{lemma}\label{lemma:existence-of-lambda}
If there exists a $\lambda^* \geq 0$ such that the optimal solution of the unconstrained problem \eqref{eqn:unconstrained-problem-for-estimate-correction} under  $\lambda^*$ meets the constraint in 
\eqref{eqn:constrained-problem-for-estimate-correction} with equality, then the optimal solution of the unconstrained problem \eqref{eqn:unconstrained-problem-for-estimate-correction} under  $\lambda^*$ is optimal for the constrained problem \eqref{eqn:constrained-problem-for-estimate-correction} as well.
\end{lemma}
The proof   is similar to that of \cite[Theorem~$1$]{arpan-globecom2017-techreport}.

\subsection{The proposed learning algorithm}
In this subsection, we propose an iterative learning algorithm to solve \eqref{eqn:constrained-problem-for-estimate-correction}. Our proposed algorithm is based on multi-timescale stochastic approximation (see \cite[Chapter~$6$]{borkar08stochastic-approximation-book}). Note that, Lemma~\ref{lemma:existence-of-lambda} suggests solving \eqref{eqn:unconstrained-problem-for-estimate-correction} for various values of $\lambda$ and then choosing a suitable $\lambda^*$ that meets the constraint in \eqref{eqn:constrained-problem-for-estimate-correction} with equality;   we   can solve \eqref{eqn:unconstrained-problem-for-estimate-correction} in an inner  loop and then vary $\lambda$ in an outer loop to converge to $\lambda^*$. This is achieved by varying $\lambda$ at a slower timescale as the iterations progress, and solving the unconstrained problems at faster timescales. 

Let us consider \eqref{eqn:unconstrained-problem-for-estimate-correction} for a fixed $\lambda$. Due to the unavailability of any closed-form expression of the cost function in \eqref{eqn:unconstrained-problem-for-estimate-correction}, direct computation of a gradient estimate w.r.t. $\mathbf{K}_t$ is not possible. However, we will minimize the cost function \eqref{eqn:unconstrained-problem-for-estimate-correction}  by iteratively learning an optimal   gain matrix $\mathbf{K}^*$ over time, via a stochastic gradient descent (SGD) algorithm.  Hence, we employ simultaneous perturbation stochastic approximation (SPSA, see \cite{spall92original-SPSA}) for this optimization problem. In SPSA, all elements of $\mathbf{K}_t$ are perturbed simultaneously by a random vector in two opposite directions,  the single stage cost function is evaluated for  these two perturbed gain matrices $\mathbf{K}_t^+$ and $\mathbf{K}_t^-$, and a noisy estimate of the gradient of the single stage cost is obtained from this. This noisy estimate of the gradient is then used in   SGD   for asymptotically  minimizing   $\mathbb{E}[c(t)]$.

The proposed SEC (secure estimation) algorithm uses three positive sequences $\{a(t)\}_{t \geq 0}$, $\{b(t)\}_{t \geq 0}$ and $\{d(t)\}_{t \geq 0}$ that satisfy the following conditions: (i) $\sum_{t=0}^{\infty} a(t)=\sum_{t=0}^{\infty} b(t)=\infty$, (ii) $\sum_{t=0}^{\infty} a^2(t)<\infty$, $\sum_{t=0}^{\infty} b^2(t)<\infty$, (iii) $\lim_{t \rightarrow \infty} d(t)=0$,  (iv) $\lim_{t \rightarrow \infty}\frac{a^2(t)}{d^2(t)} < \infty$ and (v) $\lim_{t \rightarrow \infty} \frac{b(t)}{a(t)}=0$. 

The first two conditions are standard for stochastic approximation (see \cite{borkar08stochastic-approximation-book}). The third condition ensures that the gradient estimate is asymptotically unbiased. The fourth condition is a technical condition required for the convergence of SPSA (see \cite{spall92original-SPSA}). The fifth condition ensures that, in our proposed algorithm, the $\mathbf{K}_t$ iteration using step size $a(t)$ runs at a faster timescale than the $\lambda(t)$ iteration using step size $b(t)$.

The algorithm requires a  small $\delta>0$. Let us define $\mathcal{K}:=\{\mathbf{K} \in \mathbb{R}^{q \times m}: ||\lambda_{max}(\mathbf{I}-\mathbf{K} \mathbf{C})|| \leq 1-\delta\}$ where $\lambda_{max}(\mathbf{I}-\mathbf{K} \mathbf{C})$ is the maximum eigenvalue of the matrix $(\mathbf{I}-\mathbf{K} \mathbf{C})$. The algorithm also requires a large number $l>0$.

We define the gain matrix $\mathbf{K}_{t,\mathcal{B}}^+$ which is the same as $\mathbf{K}_t^+$ except that the entries corresponding to the sensors from $\mathcal{B}^c$ are set to $0$, and {\em vice versa}   for the definition of  $\mathbf{K}_{t,\mathcal{B}^c}^-$.

Let $\mathbf{\hat{x}}_{\mathcal{B}}^+(t)$ and $\mathbf{\hat{x}}_{\mathcal{B}^c}^+(t)$ denote the estimates obtained by  using Kalman filters with {\em constant} gain matrices $\mathbf{K}_{t,\mathcal{B}}^+$ and $\mathbf{K}_{t,\mathcal{B}^c}^+$ applied to the observed sequence $\{\mathbf{y}(\tau): 1 \leq \tau \leq t\}$ over the time period $\tau=0,1,2,\cdots,t$. Let us similarly define $\mathbf{\hat{x}}_{\mathcal{B}}^-(t)$ and $\mathbf{\hat{x}}_{\mathcal{B}^c}^-(t)$. In particular,   $\mathbf{\hat{x}}_{\mathcal{B}}^+(t)$ is  computed recursively as follows for $\tau=0,1,\cdots,t-1$:

\begin{equation}
\mathbf{\hat{x}}_{\mathcal{B}}^{(\tau+1)} = \mathbf{A} \mathbf{\hat{x}}_{\mathcal{B}}^{(\tau)} +\mathbf{K}_{t, \mathcal{B}}^+ (\mathbf{y}(\tau+1)-\mathbf{C} \mathbf{A}\mathbf{\hat{x}}_{\mathcal{B}}^{(\tau)} )   \label{eqn:proxy-estimate}
\end{equation}
and setting $\mathbf{\hat{x}}_{\mathcal{B}}^+(t)=\mathbf{\hat{x}}_{\mathcal{B}}^{(t)}$.

The SEC algorithm is described below.

\vspace{2mm}
\hrule
\noindent {\bf The SEC algorithm}
 \hrule
 \vspace{0.5mm}
 \noindent {\bf Input:} $\{a(t)\}_{t \geq 0}$, $\{b(t)\}_{t \geq 0}$, $\{d(t)\}_{t \geq 0}$, $\bar{P}$, $l$,

\noindent {\bf Initialization:} $\mathbf{K}_1$, $\mathbf{P}_0$, $\lambda(0)>0$,   $\mathbf{\hat{x}}_{\mathcal{B}}^+(0)=\mathbf{\hat{x}}_{\mathcal{B}^c}^+(0)=\mathbf{\hat{x}}_{\mathcal{B}}^-(0)=\mathbf{\hat{x}}_{\mathcal{B}^c}^-(0)=\mathbf{\hat{x}}(0)$ for all   $\mathcal{B} \in 2^{\mathcal{B}}: |\mathcal{B}|=n_0$.

\noindent {\bf For $t=1,2,3,\cdots$:}

\begin{enumerate}
\item Collect the observation vector $\mathbf{y}_t$ from all sensors.
\item Compute the estimate $\mathbf{\hat{x}}(t)=\mathbf{A} \mathbf{\hat{x}}(t-1)+\mathbf{K}_t (\mathbf{y}_t- \mathbf{C} \mathbf{A} \mathbf{\hat{x}}(t-1) )$.
\item Pick a random matrix $\mathbf{\Delta}_t$ such that each entry in $\mathbf{\Delta}_t$ is chosen from $\{-1,1\}$ independently with probability $\frac{1}{2}$. 
\item Compute $\mathbf{K}_t^+:=\mathbf{K}_t+d(t) \mathbf{\Delta}_t$ and $\mathbf{K}_t^-:=\mathbf{K}_t-d(t) \mathbf{\Delta}_t$.
\item Compute $\mathbf{\hat{x}}_{\mathcal{B}}^+(t)$, $\mathbf{\hat{x}}_{\mathcal{B}^c}^+(t)$, $\mathbf{\hat{x}}_{\mathcal{B}}^-(t)$ and $\mathbf{\hat{x}}_{\mathcal{B}^c}^-(t)$ for all subsets $\mathcal{B}$ of size $n_0$, using \eqref{eqn:proxy-estimate} or similar update equations.
\item Compute $\mathbf{P}_t^+:=(\mathbf{I}-\mathbf{K}_t^+ \mathbf{C})(\mathbf{A} \mathbf{P}_{t-1}\mathbf{A}'+\mathbf{Q})(\mathbf{I}-\mathbf{K}_t^+ \mathbf{C})'+\mathbf{K}_t^+ \mathbf{R} (\mathbf{K}_t^+)'$ and similarly    $\mathbf{P}_t^-$. Compute $\mathbf{P}_t:=(\mathbf{I}-\mathbf{K}_t \mathbf{C})(\mathbf{A} \mathbf{P}_{t-1}\mathbf{A}'+\mathbf{Q})(\mathbf{I}-\mathbf{K}_t \mathbf{C})'+\mathbf{K}_t \mathbf{R} \mathbf{K}_t'$.
\item Compute $c^+(t):= \max_{\mathcal{B} \in 2^{\mathcal{N}}:|\mathcal{B}|=n_0}   ||\mathbf{\hat{x}}_{\mathcal{B}}^+(t)-\mathbf{\hat{x}}_{\mathcal{B}^c}^+(t)||^2 +\lambda(t) \mbox{Tr}(\mathbf{P}_t^+)$, and similarly   $c^-(t)$.
\item For each entry $\mathbf{K}_t(i,j)$, do the following SPSA update:
\begin{equation}
\mathbf{\tilde{K}}_{t+1}(i,j)=\bigg[\mathbf{K}_t(i,j)-a(t) \times \frac{c^+(t)-c^-(t)}{2 d(t) \mathbf{\Delta}_t(i,j)}\bigg]_{-l}^l  \label{eqn:SPSA-update-equation}
\end{equation}
and project $\mathbf{\tilde{K}}_{t+1}$ onto $\mathcal{K}$ in order to obtain $\mathbf{K}_{t+1}$.
\item Update $\lambda(t+1)=[\lambda(t)+b(t) (\mbox{Tr}(\mathbf{P}_t)-\bar{P})]_0^l$.
\end{enumerate}
\noindent {\bf end}
\label{algorithm:correction-algorithm-learning} 
 \hrule
\vspace{2mm}

Note that, $\mathbf{\tilde{K}}_{t+1}(i,j)$ is projected onto a compact interval $[-l,l]$ to ensure stability of the iteration \eqref{eqn:SPSA-update-equation}. Similarly, $\lambda(t+1)$ is projected onto a compact interval $[0,l]$. 

The spectral radius of $(\mathbf{I}-\mathbf{K}_{t+1} \mathbf{C})$ is maintained to be less than $(1-\delta)$ to ensure that the error covariance matrix $\mathbf{P}_t$ remains bounded. A standard result says  that, the covariance matrix $\mathbf{P}_t$ of the estimation error $\mathbf{\hat{x}}(t)-\mathbf{x}(t)$  varies according to the following recursive equation:

\footnotesize
\begin{equation}
\mathbf{P}_t:=(\mathbf{I}-\mathbf{K}_t \mathbf{C})(\mathbf{A} \mathbf{P}_{t-1}\mathbf{A}'+\mathbf{Q})(\mathbf{I}-\mathbf{K}_t \mathbf{C})'+\mathbf{K}_t \mathbf{R} (\mathbf{K}_t)' ,\label{eqn:evolution-of-error-covariance-for-general-gain-matrix-sequence}
\end{equation}
\normalsize
when the gain matrix $\mathbf{K}_t$ is chosen arbitrarily (not optimally as in \eqref{eqn:kalman-filter-for-single-sensor}). Step~$6$ of SEC is motivated by the above  expression.

 Equation \eqref{eqn:SPSA-update-equation} is a SGD algorithm where a noisy estimate of the gradient of $\mathbb{E}[c(t)]$ w.r.t. $\mathbf{K}_t$ is used instead of the true gradient. The noisy gradient estimate is $\frac{c^+(t)-c^-(t)}{2 b(t) \mathbf{\Delta}_t(i,j)}$. Note that, SPSA requires  computation of $c^+(t)$ and $c^-(t)$ only for two different gain matrices $\mathbf{K}_t^+$ and $\mathbf{K}_t^-$; this allows us to avoid unnecessarily huge computation involved in gradient estimation using coordinatewise perturbation.

\subsection{Convergence of SEC}
Let us consider the problem in \eqref{eqn:unconstrained-problem-for-estimate-correction}. We  define  a function $P(\mathbf{K})$ which is the time-averaged MSE (and also the limiting MSE) $\lim_{T \rightarrow \infty} \frac{1}{T}\sum_{t=1}^T \mathbb{E} [\mbox{Tr} (\mathbf{P}_t)]$ achieved if a Kalman-like linear estimator  is used with a constant gain matrix $\mathbf{K}$ for all $t$, when there is no attack. Also, for a constant gain matrix $\mathbf{K}$ used at all time $t$, let us define by $g(\mathbf{K})$ the limiting (and also time-average) value of the anomaly $\limsup_{t \rightarrow \infty} \mathbb{E} [ \max_{\mathcal{B} \in 2^{\mathcal{N}}:|\mathcal{B}|=n_0}||\mathbf{\hat{x}}_{\mathcal{B}}(t)-\mathbf{\hat{x}}_{\mathcal{B}^c}(t)||^2 ]$, when there is a possible attack. We define $C(\mathbf{K},\lambda):=g(\mathbf{K})+\lambda P(\mathbf{K})$.  
Let us also define $\mathcal{K}_{\lambda}:=\{\mathbf{K} \in [-l,l]^{q \times m} \cap \mathcal{K}: \nabla_{\mathbf{K}} C(\mathbf{K},\lambda)=0 \}$. 

\begin{assumption}\label{assumption:lipschitz-continuity}
For any fixed $\lambda \in [0,l]$,  the function $C(\mathbf{K},\lambda)$ is Lipschitz continuous in $\mathbf{K} \in [-l,l]^{q \times m} \cap \mathcal{K}$.
\end{assumption}

Let us recall that, $$\mathbf{P}_t:=(\mathbf{I}-\mathbf{K}_t \mathbf{C})(\mathbf{A} \mathbf{P}_{t-1}\mathbf{A}'+\mathbf{Q})(\mathbf{I}-\mathbf{K}_t \mathbf{C})'+\mathbf{K}_t \mathbf{R} (\mathbf{K}_t)'$$
 Since the spectral radius of $(\mathbf{I}-\mathbf{K}_t \mathbf{C})$ is less than or equal to $(1-\delta)<1$, the $\mathbf{P}_t$ iteration converges almost surely, and hence the MSE under SEC is uniformly bounded across sample paths. If a constant gain matrix $\mathbf{K}$ is used, it is still easy to prove that $P(\mathbf{K})$ is Lipschitz continuous in $\mathbf{K} \in [-l,l]^{q \times m} \cap \mathcal{K}$. Thus, Assumption~\ref{assumption:lipschitz-continuity} is specifically required for $g(\mathbf{K})$.

\begin{theorem}\label{theorem:convergence-of-SPSA}
Under SEC with a fixed $\lambda(t)=\lambda$ (i.e., $b(t)=0$ for all $t=0,1,2,\cdots$)  and Assumption~\ref{assumption:lipschitz-continuity}, the iterates $\{\mathbf{K}_t\}_{t \geq 1}$ converge almost surely to the set $\mathcal{K}_{\lambda}$, provided that each such stationary point belongs to the interior of 
$[-l,l]^{q \times m} \cap \mathcal{K}$.
\end{theorem}
\begin{proof}
 See Appendix~\ref{appendix:proof-of-convergence-of-SPSA}.
\end{proof}
Theorem~\ref{theorem:convergence-of-SPSA} says that $\{\mathbf{K}_t\}$ converges to the set of local minima of $C(\mathbf{K})$ (i.e. $C(\mathbf{K},\lambda)$ for a given $\mathbf{K}$) in case there is no saddle point that is not a local minimum.

However, SEC varies $\lambda(t)$ at a slower timescale in order to solve the constrained problem~\eqref{eqn:constrained-problem-for-estimate-correction}. The next theorem provides our main result on  the convergence of SEC  in its original form, to the desired solution set  of the constrained problem~\eqref{eqn:constrained-problem-for-estimate-correction}.

Let $\mathcal{S}_{\lambda}$ denote the closure of the convex hull of the set $\{P(\mathbf{K}) -\bar{P}: \mathbf{K} \in \mathcal{K}_{\lambda} \}$. We define by $\bar{\Lambda}$ the collection of  closed connected internally chain transitive invariant sets  of the following differential inclusion (reference  \cite{smirnov2002introduction}):
$$\dot{\lambda}(\tau)\in \mathcal{S}_{\lambda(\tau)}, \tau \in \mathbb{R}^+$$

\begin{theorem}\label{theorem:convergence-of-SPSA-varying-Lagrange-multiplier}
Under SEC and Assumption~\ref{assumption:lipschitz-continuity},  the sequence  $\{(\mathbf{K}_t,\lambda(t))\}_{t \geq 1}$ almost surely converges to the set $ \{(\mathbf{K},\lambda): \mathbf{K} \in \mathcal{K}_{\lambda}, \lambda \in \bar{\Lambda}\}$.
\end{theorem}
\begin{proof}
See Appendix~\ref{appendix:proof-of-convergence-of-SPSA-varying-Lagrange-multiplier}.
\end{proof}
The proof of Theorem~\ref{theorem:convergence-of-SPSA-varying-Lagrange-multiplier} suggests that the $\lambda(t)$ update equation asymptotically behaves like a stochastic recursive inclusion (see \cite[Chapter~$5$]{borkar08stochastic-approximation-book}), where at each time, the iteration  involves a set $\{P(\mathbf{K})-\bar{P}: \mathbf{K} \in \mathcal{K}_{\lambda(t)} \}$. However, the proof requires that the set $\{P(\mathbf{K})-\bar{P}: \mathbf{K} \in \mathcal{K}_{\lambda(t)} \}$ should be convex and compact, which may not be true in general. Hence, we consider the set $\mathcal{S}_{\lambda(t)}$, which results in a weaker result on the convergence of SEC in Theorem~\ref{theorem:convergence-of-SPSA-varying-Lagrange-multiplier}.

\subsection{Complexity issues and a heuristic}\label{subsection:complexity-and-heuristic-secure-estimation}
The SEC algorithm requires us to compute $\mathbf{P}_t,\mathbf{P}_t^+, \mathbf{P}_t^-$ and also $\max_{\mathcal{B} \in 2^{\mathcal{N}}:|\mathcal{B}|=n_0}   ||\mathbf{\hat{x}}_{\mathcal{B}}^+(t)-\mathbf{\hat{x}}_{\mathcal{B}^c}^+(t)||^2 $ and $\max_{\mathcal{B} \in 2^{\mathcal{N}}:|\mathcal{B}|=n_0}   ||\mathbf{\hat{x}}_{\mathcal{B}}^-(t)-\mathbf{\hat{x}}_{\mathcal{B}^c}^-(t)||^2 $. The matrix $\mathbf{P}_t$ can be computed iteratively using \eqref{eqn:evolution-of-error-covariance-for-general-gain-matrix-sequence}. However, ideally we should use $P(\mathbf{K}_t^+)$ instead of $\mathbf{P}_t^+$, and $P(\mathbf{K}_t^-)$ instead of $\mathbf{P}_t^-$, in the SPSA update \eqref{eqn:SPSA-update-equation}. Computing $P(\mathbf{K}_t^+)$ will require us to run the iteration \eqref{eqn:evolution-of-error-covariance-for-general-gain-matrix-sequence} $t$ times at iteration $t$; i.e., the computational complexity of $P(\mathbf{K}_t^+)$ grows with $t$. But, in the proof of Theorem~\ref{theorem:convergence-of-SPSA}, we can show that $\mathbf{P}_t^+=P(\mathbf{K}_t^+)+o(d(t))$ and $\mathbf{P}_t^-=P(\mathbf{K}_t^-)+o(d(t))$. This allows us to use $\mathbf{P}_t^+$ and $\mathbf{P}_t^-$ matrices in \eqref{eqn:SPSA-update-equation}, whose computation can be done recursively as  in step~$6$ of SEC. Thus, \eqref{eqn:SPSA-update-equation} is not a standard SPSA acheme.

However, the computation of $\mathbf{\hat{x}}_{\mathcal{B}}^+(t), \mathbf{\hat{x}}_{\mathcal{B}^c}^+(t), \mathbf{\hat{x}}_{\mathcal{B}}^-(t), \mathbf{\hat{x}}_{\mathcal{B}^c}^-(t)$ for all $\mathcal{B}$ of cardinality $n_0$ requires $O(t)$ computation at time $t$; this happens because at each time $t$ we obtain two new matrices $\mathbf{K}_t^+$ and $\mathbf{K}_t^-$, and the estimates $\mathbf{\hat{x}}_{\mathcal{B}}^+(t), \mathbf{\hat{x}}_{\mathcal{B}^c}^+(t), \mathbf{\hat{x}}_{\mathcal{B}}^-(t), \mathbf{\hat{x}}_{\mathcal{B}^c}^-(t)$ for all $\mathcal{B}$ need to be computed using constant gain matrices $\mathbf{K}_t^+$ and $\mathbf{K}_t^-$ applied to the histrory of observations $\{\mathbf{y}(\tau): 1 \leq \tau \leq t\}$. This restricts the possibility of using SEC in practical applications, since   IoT applications will require low-complexity solutions.

One possible heuristic to alleviate this problem is to update $\mathbf{K}_t$ only up to some fixed $T$ time steps, and afterwards use this constant $\mathbf{K}_t$ matrix for estimation for ever. However, this results in the loss of the very essence of SEC. SEC does observation driven gain adjustment to tackle   FDI attacks; if an attack starts beyond time $t=T$, this modified security algorithm will miss the attack. Another possible strategy could be to run the SPSA sequence for $T$ time steps, and then repeat the procedure for time $t \in \{iT+1,T+2,\cdots, (i+1)T\}$ with the same step size sequence $a(1), a(2),\cdots,a(T)$, but on the observation sequence $\mathbf{y}(iT+1),\mathbf{y}(iT+2),\cdots, \mathbf{y}((i+1)T)$. This can at best ensure convergence within a neighbourhood of the solution set. However, choice of a large $T$ still results in large computational complexity.

Now we present an alternative low-complexity version of SEC called SEC-L, which recursively computes $\mathbf{\tilde{x}}_{\mathcal{B}}^+(t), \mathbf{\tilde{x}}_{\mathcal{B}^c}^+(t), \mathbf{\tilde{x}}_{\mathcal{B}}^-(t), \mathbf{\tilde{x}}_{\mathcal{B}^c}^-(t)$ for all $\mathcal{B}$ of cardinality $n_0$, which are (suboptimal) proxies for $\mathbf{\hat{x}}_{\mathcal{B}}^+(t), \mathbf{\hat{x}}_{\mathcal{B}^c}^+(t), \mathbf{\hat{x}}_{\mathcal{B}}^-(t), \mathbf{\hat{x}}_{\mathcal{B}^c}^-(t)$. 

{\bf SEC-L algorithm:} This algorithm is same as SEC, except that, at step~$5$, we compute: $\mathbf{\tilde{x}}_{\mathcal{B}}^+(t)=\mathbf{A} \mathbf{\hat{x}}(t-1)+\mathbf{K}_{t,\mathcal{B}}^+(\mathbf{y}_t-\mathbf{C}\mathbf{A}\mathbf{\hat{x}}(t-1))$. The estimates $\mathbf{\tilde{x}}_{\mathcal{B}}^-(t), \mathbf{\tilde{x}}_{\mathcal{B}^c}^+(t), \mathbf{\tilde{x}}_{\mathcal{B}^c}^-(t)$ for all subsets $\mathcal{B}$ of size $n_0$ are also calculated similarly.\qed

Clearly, SEC-L  reduces the $O(t)$ complexity for computing $\mathbf{\hat{x}}_{\mathcal{B}}^+(t), \mathbf{\hat{x}}_{\mathcal{B}^c}^+(t), \mathbf{\hat{x}}_{\mathcal{B}}^-(t)$ and  $\mathbf{\hat{x}}_{\mathcal{B}^c}^-(t)$  to $O(1)$, since SEC-L does not involve filtering over the entire observation history.

\subsection{Packet loss}\label{subsection:extension-to-lossy-links-secure-estimation}
So far we have assumed that all observations reach the fusion center without error. But, in practice, the links between a sensor and the fusion center can be unreliable, and hence some of the  observation packets sent from the sensors to the fusion center might be lost. Let us assume that the observation packet sent by sensor~$i$ at time $t$ is lost with a known probability $p_i$; packet loss is assumed to be i.i.d. across time $t$ and independent across sensors. In this case, at each time, the fusion center will receive an observation vector of variable size depending on the lost observation packets. However, since the fusion center knows the sensors from which observations are received at the current time step, the fusion center can simply restrict  the observation matrix $C$ at time $t$ to the set of observed sensors, and update only that part of $\mathbf{K}_t$ (via SPSA) which  corresponds to the observed sensors. However, since various submatrices of the gain matrix are updated at various time instants, we need to update them using asynchronous stochastic approximation (\cite[Chapter~$7$]{borkar08stochastic-approximation-book}) in the SPSA step. This requires us to maintain a counter $\nu_S(t)$ for each sensor subset $\mathcal{S} \in 2^{\mathcal{N}}$; $\nu_{\mathcal{S}}(t)$ is the number of times till $t$ when observations came only from sensor subset $\mathcal{S}$. Now, SEC can be adapted to this case in the following way. At time $t$, let the observed sensor subset be $\mathcal{S}_t$. Then, estimation is done by the following rule:
$$\mathbf{\hat{x}}(t)=\mathbf{A} \mathbf{\hat{x}}(t-1)+\mathbf{K}_{t,\mathcal{S}_t}(\mathbf{y}_t- \mathbf{C}_{\mathcal{S}_t} \mathbf{A} \mathbf{\hat{x}}(t-1))$$
and the $\mathbf{K}_t$ update equation is modified as:
$$ \mathbf{\tilde{K}}_{t+1,\mathcal{S}_t}(i,j)=\bigg[\mathbf{K}_{t,\mathcal{S}_t}(i,j)-a(\nu_{\mathcal{S}}(t))\frac{c^+(t)-c^-(t)}{2 d(t) \mathbf{\Delta}_t(i,j)}\bigg]_{-l}^l  $$
where $\mathbf{K}_{t,\mathcal{S}_t}$ is the restriction of $\mathbf{K}_t$ to the sensor subset $\mathcal{S}_t$, and $\mathbf{C}_{\mathcal{S}_t}$ is the restriction of $\mathbf{C}$ to subset $\mathcal{S}_t$. It is easy to adapt our convergence proofs for Theorem~\ref{theorem:convergence-of-SPSA} and Theorem~\ref{theorem:convergence-of-SPSA-varying-Lagrange-multiplier} to this modified algorithm.

\section{Attack detection}\label{section:attack-detection}
In this section, we  develop an efficient detection algorithm for the FDI attack  on  a {\em static} sensor subset, though the proposed detector can be heuristically used to detect  a general stationary attack. This algorithm is only meant for attack detection, with the assumption that necessary measures will be taken if an attack is detected.   The detection problem is mathematically represented as a hypothesis testing problem on the two hypotheses:
\begin{itemize}
\item $\mathcal{H}_0$: there is no attack; $\mathbf{e}_i(t)=0 \,\, \forall \,\,i \in \mathcal{N}$, $t=1,2,\cdots$
\item $\mathcal{H}_1$: there is an attack;  $\mathbf{e}_i(t) \neq 0$ for some  $i \in \mathcal{N}$
\end{itemize}
with observation sequence $\{\mathbf{y}(t)\}_{t \geq 1}$. 
Note that, due to the complicated dynamics involved in Kalman filtering, it is difficult to carry out standard hypothesis testing schemes. Also, due to the unavailability of any known safe sensor, we cannot compare the innovation sequence against any reliable quantity. However, if a subset of sensors is under attack, then the process estimate obtained only from these sensor observations is likely to have high error, and hence should be significantly different from the estimates made by other sensor observations. We exploit this fact to develop an attack detector. 

Let us denote by $\mathbf{\hat{x}}_{\mathcal{B}}(t)$ the process estimate returned by an {\em optimal} Kalman filter that uses observations made by the sensor subset $\mathcal{B}$ only (see \eqref{eqn:kalman-filter-for-single-sensor}). Let us denote the covariance matrix of the anomaly $\mathbf{e}_{\mathcal{B},\mathcal{B}^c}(t):=(\mathbf{\hat{x}}_{\mathcal{B}}(t)-\mathbf{\hat{x}}_{\mathcal{B}^c}(t))$ by $\mathbf{\underbar{P}}_{\mathcal{B},\mathcal{B}^c}$ under steady state, when there is no attack.  Clearly, if there is no attack, then, under steady state, $\mathbf{e}_{\mathcal{B},\mathcal{B}^c}(t) \sim N(\mathbf{0}, \mathbf{\underbar{P}}_{\mathcal{B}, \mathcal{B}^c})$  since the error $(\mathbf{\hat{x}}_{\mathcal{B}}(t)-\mathbf{x}(t))$ and  
$(\mathbf{\hat{x}}_{\mathcal{B}^c}(t)-\mathbf{x}(t))$ are zero-mean Gaussian. Hence, one can detect an attack by checking whether $\mathbf{e}_{\mathcal{B},\mathcal{B}^c}(t)$ is coming from the distribution $N(\mathbf{0}, \mathbf{\underbar{P}}_{\mathcal{B}, \mathcal{B}^c})$ for each subset $\mathcal{B}$ of size $n_0$. The covariance matrix $\mathbf{\underbar{P}}_{\mathcal{B},\mathcal{B}^c}$ can be pre-computed by simulating the process beforehand.

The algorithm requires a positive integer $J$ as an observation window, and a threshold $\eta>0$ for attack detection.

The algorithm to detect and localize an attack is given below. We call this algorithm DETECT.

\hrule
\noindent {\bf The DETECT algorithm}
\hrule
\noindent {\bf Input:} $J$, $\eta$.\\
\noindent {\bf Off-line pre-computation:} 
Simulate  $\{x(\tau)\}_{\tau \geq 0}$   off-line using \eqref{eqn:process-equation}. In this simulated process, compute $\{\mathbf{\hat{x}}_{\mathcal{B}}(\tau)\}_{\tau \geq 0}$ and $\{\mathbf{\hat{x}}_{\mathcal{B}^c}(\tau)\}_{\tau \geq 0}$ via optimal Kalman filtering for  each sensor subset $\mathcal{B}$ of size $n_0$, by suitable adaptation of  \eqref{eqn:kalman-filter-for-single-sensor}. Compute  $\mathbf{\underbar{P}}_{\mathcal{B}, \mathcal{B}^c}:=\lim_{T \rightarrow \infty} \frac{1}{T}\sum_{\tau=1}^{T} \mathbf{e}_{\mathcal{B},\mathcal{B}^c}(\tau)  (\mathbf{e}_{\mathcal{B},\mathcal{B}^c}(\tau) )'$. \\

\noindent {\bf Attack detection in the physical process:} 

\noindent {\bf For $t=1,2,\cdots$}

 \begin{enumerate}
\item Use the optimal Kalman filter \eqref{eqn:kalman-filter-for-single-sensor} to compute $\mathbf{\hat{x}}(t)$.
\item Compute $\{\mathbf{\hat{x}}_{\mathcal{B}}(t)\}_{t \geq 0}$ and $\{\mathbf{\hat{x}}_{\mathcal{B}^c}(t)\}_{t \geq 0}$ via optimal Kalman filtering for  each sensor subset $\mathcal{B}$ of size $n_0$, by suitable adaptation of  \eqref{eqn:kalman-filter-for-single-sensor}.  Compute $\mathbf{e}_{\mathcal{B},\mathcal{B}^c}(t)=(\mathbf{\hat{x}}_{\mathcal{B}}(t)-\mathbf{\hat{x}}_{\mathcal{B}^c}(t))$ for  each sensor subset $\mathcal{B}$ of size $n_0$.

\item Declare that an attack has happened if: 
 $$\max_{\{\mathcal{B} \in 2^{\mathcal{N}}: |\mathcal{B}|=n_0\}} \sum_{\tau=t-J+1}^t (\mathbf{e}_{\mathcal{B},\mathcal{B}^c}(\tau))' \mathbf{\underbar{P}}_{\mathcal{B}, \mathcal{B}^c}^{-1} \mathbf{e}_{\mathcal{B},\mathcal{B}^c}(\tau) >\eta$$

\item  If an attack is declared, identify the following maximizing subset as the attacked sensor subset:  
$$\arg \max_{\mathcal{B} \in 2^{\mathcal{N}}: |\mathcal{B}|=n_0} \sum_{\tau=t-J+1}^t (\mathbf{e}_{\mathcal{B},\mathcal{B}^c}(\tau))' \mathbf{\underbar{P}}_{\mathcal{B}, \mathcal{B}^c}^{-1} \mathbf{e}_{\mathcal{B},\mathcal{B}^c}(\tau)$$
\end{enumerate}
\noindent {\bf end}
\hrule

 The detection step is similar to the standard $\chi^2$ test used to check whether a sequence of random vectors are coming from a desired Gaussian distribution, except that this test is conducted on all possible subsets of size $n_0$, and hence the $\max$ operation is needed.
 The false alarm probability can be controlled via selection of the threshold $\eta$.

Now we provide a learning scheme to find the optimal $\eta$ for a given target on the false alarm probability. 
The false alarm probability $P_F$ under DETECT is defined as:

\footnotesize
\begin{equation*}
P_F 
=\lim_{t \rightarrow \infty}  \mathbb{P}\bigg( \max_{\mathcal{B} \in 2^{\mathcal{N}}: |\mathcal{B}|=n_0} \sum_{\tau=t-J}^t (\mathbf{e}_{\mathcal{B},\mathcal{B}^c}(\tau))' 
 \mathbf{\underbar{P}}_{\mathcal{B}, \mathcal{B}^c}^{-1} \mathbf{e}_{\mathcal{B},\mathcal{B}^c}(\tau) >\eta \bigg | \mathcal{H}_0 \bigg)
\end{equation*}
\normalsize

In order to satisfy the constraint $P_F \leq \alpha$ with equality, we need to choose an optimal threshold $\eta_{\alpha}^*$ in DETECT. The optimal $\eta_{\alpha}^*$ can be computed by using the following LEARN algorithm (a stochastic approximation step) in the off-line pre-computation phase of DETECT.

The LEARN algorithm requires a positive sequence $\{a(\tau)\}_{\tau \geq 0}$ such that $\sum_{\tau =0}^{\infty} a(\tau)=\infty$ and $\sum_{\tau=0}^{\infty} a^2(\tau)<\infty$. LEARN also requires $\mathbf{\underbar{P}}_{\mathcal{B},\mathcal{B}^c}$ for all subset $\mathcal{B}$ of size $n_0$ as input. The algorithm simulates the $\mathbf{x}(\tau)$ process  off-line, and maintains a   detector as in DETECT with an initial threshold $\eta(0)$. Let us denote the number of false alarm triggers made by this detector up to time $(\tau-1)$ in the simulated process by $N_{\tau-1}$.

\hrule
\hrule 
\noindent {\bf The LEARN algorithm}
\hrule
\hrule
\noindent {\bf Input:} $J$, $\{a(\tau)\}_{\tau \geq 0}$, $\mathbf{\underbar{P}}_{\mathcal{B},\mathcal{B}^c}$ for all subsets $\mathcal{B}$ of size $n_0$.
{\bf Initialization:} $\eta(0)$, $N_0=0$

\noindent {\bf For $\tau=0,1,2,\cdots$ in the simulated process:}
\begin{enumerate}
\item  Check if 
$\max_{\mathcal{B} \in 2^{\mathcal{N}}: |\mathcal{B}|=n_0} \sum_{n=\tau-J+1}^{\tau} (\mathbf{e}_{\mathcal{B},\mathcal{B}^c}(n))' \mathbf{\underbar{P}}_{\mathcal{B}, \mathcal{B}^c}^{-1} \mathbf{e}_{\mathcal{B},\mathcal{B}^c}(n) >\eta(\tau)$.
\item If this condition is satisfied, update $N_{\tau}=N_{\tau-1}+1$, else $N_{\tau}=N_{\tau-1}$.
\item Update the threshold  $\eta(\tau+1)=[\eta(\tau)+a(\tau) (\mathbb{I}(N_{\tau}>N_{\tau-1})-\alpha)]_0^l$.
\end{enumerate}
\noindent {\bf end}

\hrule
\hrule

The $\eta(\tau)$ update scheme is a stochastic approximation algorithm (see \cite{borkar08stochastic-approximation-book}).
The goal of the $\eta(\tau)$ update scheme is to meet the false alarm probability constraint with equality. If a false alarm is triggered  at time $\tau$ in the simulation, $\eta(\tau)$ is increased; else, $\eta(\tau)$ is decreased. By the theory of \cite{borkar08stochastic-approximation-book}, it is straightforward to show that $\lim_{\tau \rightarrow \infty} \eta(\tau)=\eta_{\alpha}^*$ almost surely, and $\lim_{\tau \rightarrow \infty} \mathbb{P} (n_{\tau+1}>n_{\tau})=\alpha$.
 $l$ is a large positive number such that  $\eta_{\alpha}^* \in (0,l)$. The projection operation is used to ensure boundedness of the $\eta(\tau)$ iterates.


\section{Numerical results}\label{section:numerical-work}
In this section, we numerically demonstrate the efficacy of SEC-L and DETECT. For attack detection, we compare the performance of DETECT with the traditional $\chi^2$ detector, and also with the detector of \cite{li2017detection}.  The algorithm of \cite{li2017detection} assumes the availability of a set of {\em safe sensors} (SAFE). In SAFE, at each time $t$, observations are collected from all sensors, but the safe sensor observations are used by  the Kalman filter to generate an initial estimate. Then, the observations from potentially unsafe sensors are passed through a $\chi^2$ detector, and those observations are used in a Kalman filter to obtain   $\mathbf{\hat{x}}(t)$ if and only if the $\chi^2$ detector is not triggered.  
 
 For secure estimation, we also compare the performance of SEC-L with a blind Kalman filter oblivious to cyber-attack (KALMAN), and a Kalman filter which perfectly knows (genie-aided) the malicious sensors and   ignores their observations (we call this estimator GENIE). Note that, we do not investigate the performance for the original SEC algorithm in order to avoid the huge computational burden that grows with time, but we recall that SEC-L is motivated by SEC.

In each case, we consider an independent realization of a system with the following parameters.  The state transition matrix $\mathbf{A}$ is taken as a randomly generated $q \times q$ stochastic matrix multiplied by $0.5$. State noise covariance matrix $\mathbf{Q} \in \mathbb{R}^{q \times q}$ is chosen to be a positive semidefinite (PSD) matrix such that $\mathbf{Q}^{\frac{1}{2}}=0.1 \mathbf{Z}$ where $\mathbf{Z}_{i,j} \sim u[-1,1]$. The  matrix $\mathbf{R}$ is also chosen similarly. Observation matrix $\mathbf{C} \in \mathbb{R}^{kN \times q}$ is chosen randomly, with  $\mathbf{C}_{i,j} \sim u [0,1]$;   the observation made by each sensor is a $k$-dimensional column vector. We also assume that at most $n_0$ sensors can be  attacked at a time. The attacker inverts the sign of the innovation vectors  coming from the malicious sensors.

We consider two situations: (i) the attacker knows the estimate made by the remote estimator, and (ii) the attacker can only run a Kalman filter in order to guess the estimate.  Let us denote a matrix $\mathbf{C}_a(t)$ which is same as $\mathbf{C}$, except that the entries corresponding to the benign (not malicious) sensors at time $t$  are set to $0$. Similarly, let $\mathbf{y}_a(t)$ be same as $\mathbf{y}(t)$, except that the entries corresponding to the benign   sensors  are  $0$.  

When the attacker knows the estimate made by the remote estimator, the received observation at time $t$ at the remote estimator becomes $\mathbf{\tilde{y}}(t)=\mathbf{y}(t)+2\mathbf{C}_a(t) \mathbf{A}  \mathbf{\hat{x}}(t-1)-2\mathbf{y}_a(t)$; this is equivalent to inverting the sign of the innovation vector. We call the corresponding variants of SEC-L, KALMAN and GENIE by SEC-L-K, KALMAN-K and GENIE-K (with {\em knowledge} of the estimate).

However, if $\mathbf{\hat{x}} (t-1)$ is not known to the attacker, then the attacker can run a Kalman filter on the received observations at the estimator, in order to maintain a proxy $\mathbf{\hat{x}}_{kalman}(t-1)$ for $\mathbf{\hat{x}}(t-1)$. In this case, the the received observation at time $t$ at the remote estimator is  $\mathbf{\tilde{y}}(t)=\mathbf{y}(t)+2\mathbf{C}_a(t)\mathbf{A} \mathbf{\hat{x}}_{kalman}(t-1)-2\mathbf{y}_a(t)$. We call the corresponding variants of SEC-L, KALMAN and GENIE by SEC-L-NK, KALMAN-NK and GENIE-NK ({\em no knowledge} of estimate).

 \begin{figure}[t!]
\begin{centering}
\begin{center}
\includegraphics[height=6cm, width=8cm]{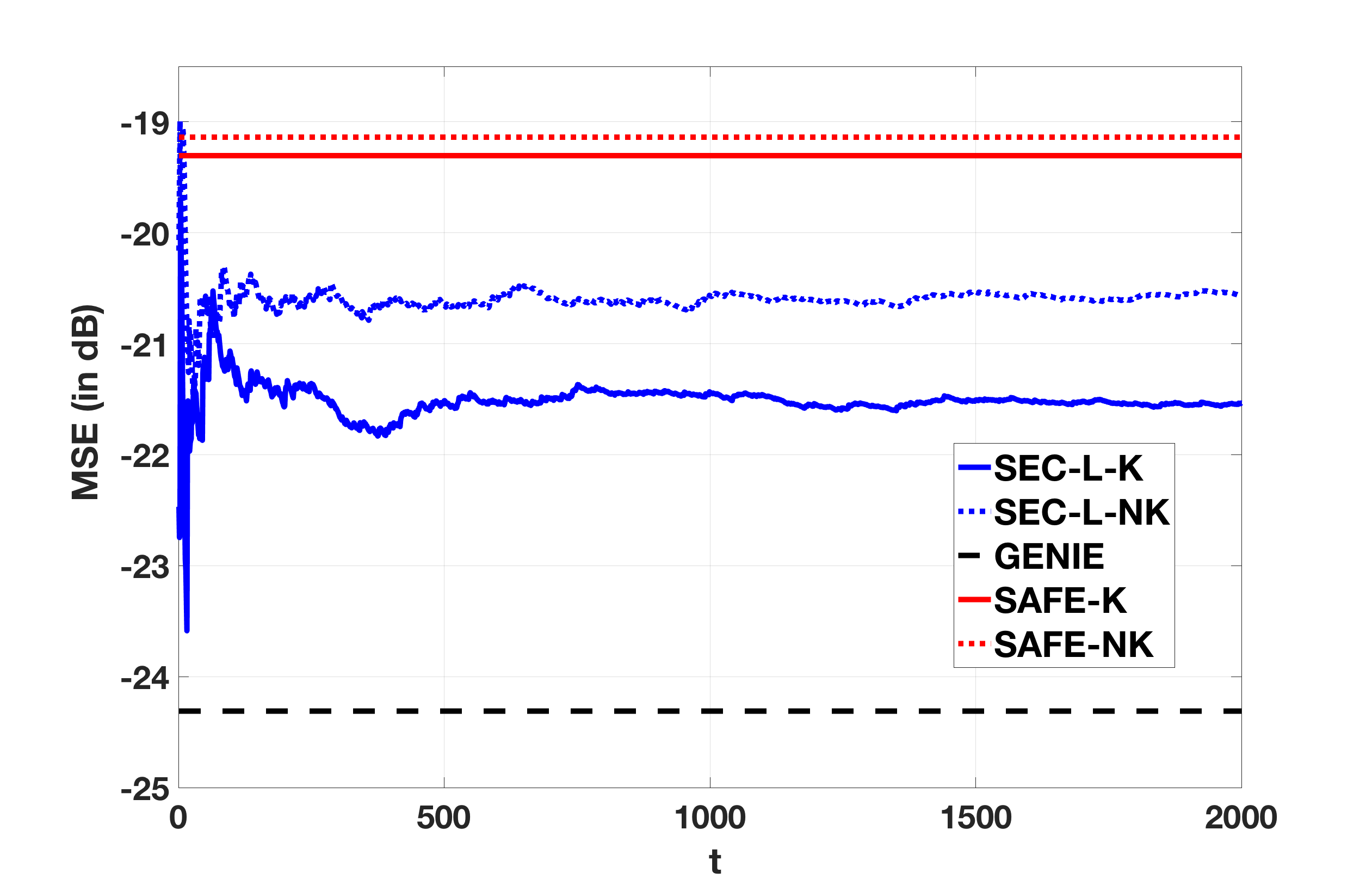}
\end{center}
\end{centering}
\caption{Performance   of SEC-L   under static attack.  $N=5$, $n_0=2$, $k=2$, $q=2$, $\lambda=2$, $a(t)=\frac{1}{2t}$, $d(t)=\frac{0.1}{t^{0.1}}$.}
\label{fig:secure-estimation-static-attack}
\end{figure}

 \begin{figure}[t!]
\begin{centering}
\begin{center}
\includegraphics[height=6cm, width=8cm]{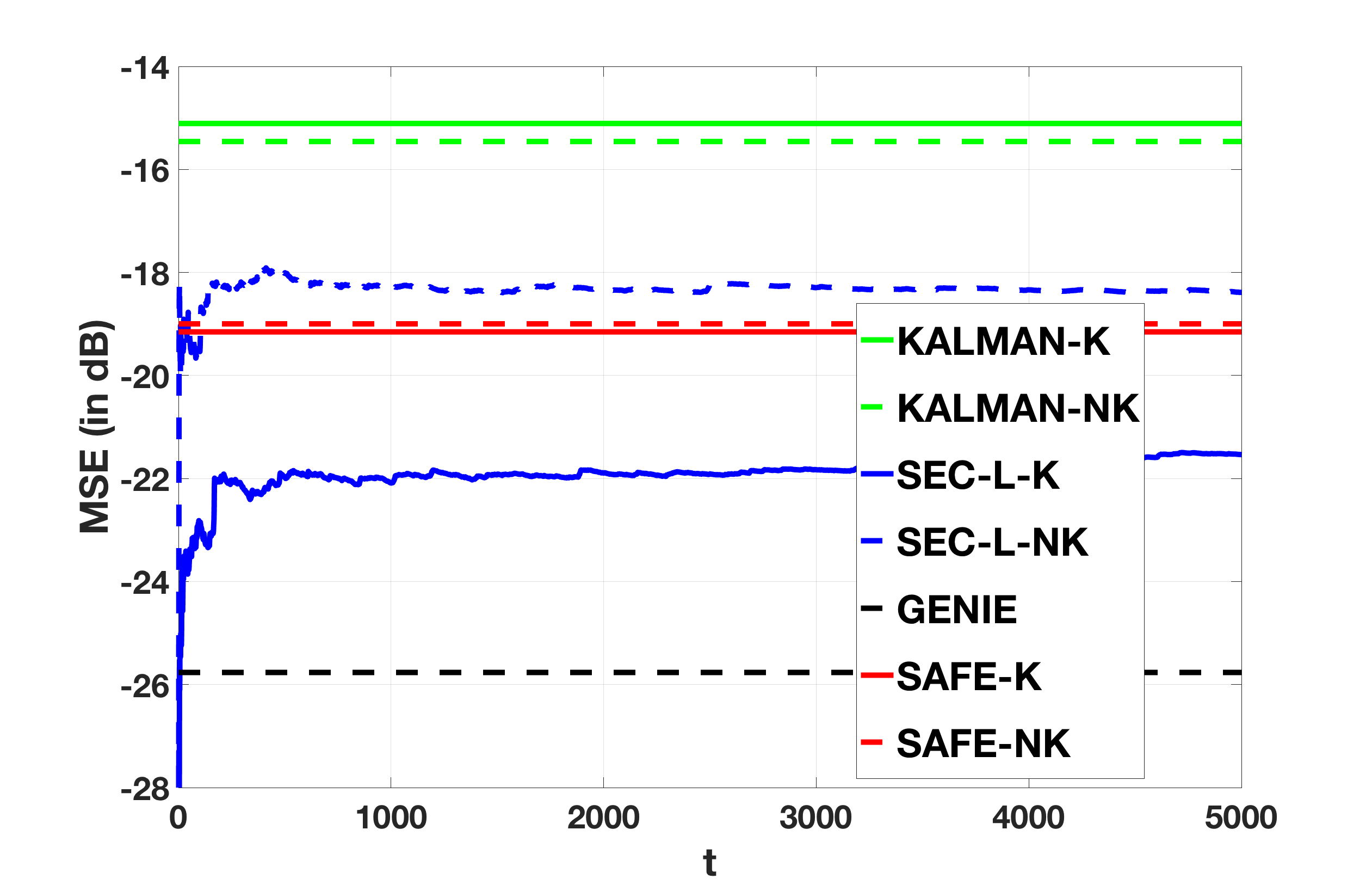}
\end{center}
\end{centering}
\caption{Performance   of SEC-L  under switching location attack.  $N=5$, $n_0=2$, $k=2$, $q=2$, $\lambda=2$, $a(t)=\frac{1}{2t}$, $d(t)=\frac{0.1}{t^{0.1}}$, $T=20$.}
\label{fig:secure-estimation-switching-attack}
\vspace{-0cm}
\end{figure}

\subsection{Secure estimation}\label{subsection:numerical-secure-estimation}
Here we   compare the time-averages MSE of SEC-L for  a fixed $\lambda=2$ with    GENIE, KALMAN and SAFE. 
\subsubsection{Secure estimation under static attack}
We consider a static attack, i.e., the attacked sensor subset $\mathcal{A}(t)$ does not vary with time. We ran the simulation for a number of independent problem instances, and computed time-averaged MSE of the various algorithms along a single sample path. We observe from Figure~\ref{fig:secure-estimation-static-attack}, that the time-average MSE of SEC-L converges within $3.5$~dB of that of  GENIE,  and is smaller than SAFE by a margin of $2$ to $3$~dB, for both K and NK situations.  However, we have observed in some other simulation instances that SEC-L can yield a larger MSE than SAFE, but without using the knowledge of any safe sensor.   SEC-L has an MSE much smaller that KALMAN, hence we do not show the performance of KALMAN. 
 
 We would like to mention here that, GENIE-K and GENIE-NK have the same performance since both ignore the observations from attacked sensors. Also, we can not order the MSE of SEC-L-K and SEC-L-NK, or SAFE-K and SAFE-NK, since SEC-L and SAFE are not provably optimal algorithms to minimize MSE under such attacks.

\subsubsection{Secure estimation under switching location attack}
We next consider an attack model where, at time instances $t=1,T+1,2T+1,\cdots$ (with $T=20$), a random sensor subset of size $n_0$ is chosen in an i.i.d. fashion, and this subset is attacked over the next $T$ slots by inverting its innovation sequence. We assume that the probability of attacking a sensor~$i$ is proportional to $\frac{1}{i^2}$. Since each sensor is susceptible to an FDI attack, SAFE  is not applicable here due to the lack of availability of any {\em safe} sensor; but still we compare SAFE-K and SAFE-NK with SEC-L-K, SEC-L-NK (for $\lambda=2$), KALMAN-K, KALMAN-NK and  GENIE.  We   observe from Figure~\ref{fig:secure-estimation-switching-attack} that (i) SEC-L-K has $3$~dB lower MSE  than SAFE-K, and SEC-L-NK has very similar performance to SAFE-NK, (ii) KALMAN offers the worse error performance among all algorithms, and (iii)  SEC-L yields $4$ to $7$~dB worse MSE than  GENIE. Since     GENIE always knows $\mathcal{A}(t)$ perfectly, GENIE enjoys significant advangate over SEC-L. 

As $T$ increases, the attack will become   more static in nature, and hence SEC-L will be able to adjust $\mathbf{K}_t$ more efficiently.

 \begin{figure}[t!]
\begin{centering}
\begin{center}
\includegraphics[height=3.5cm, width=6cm]{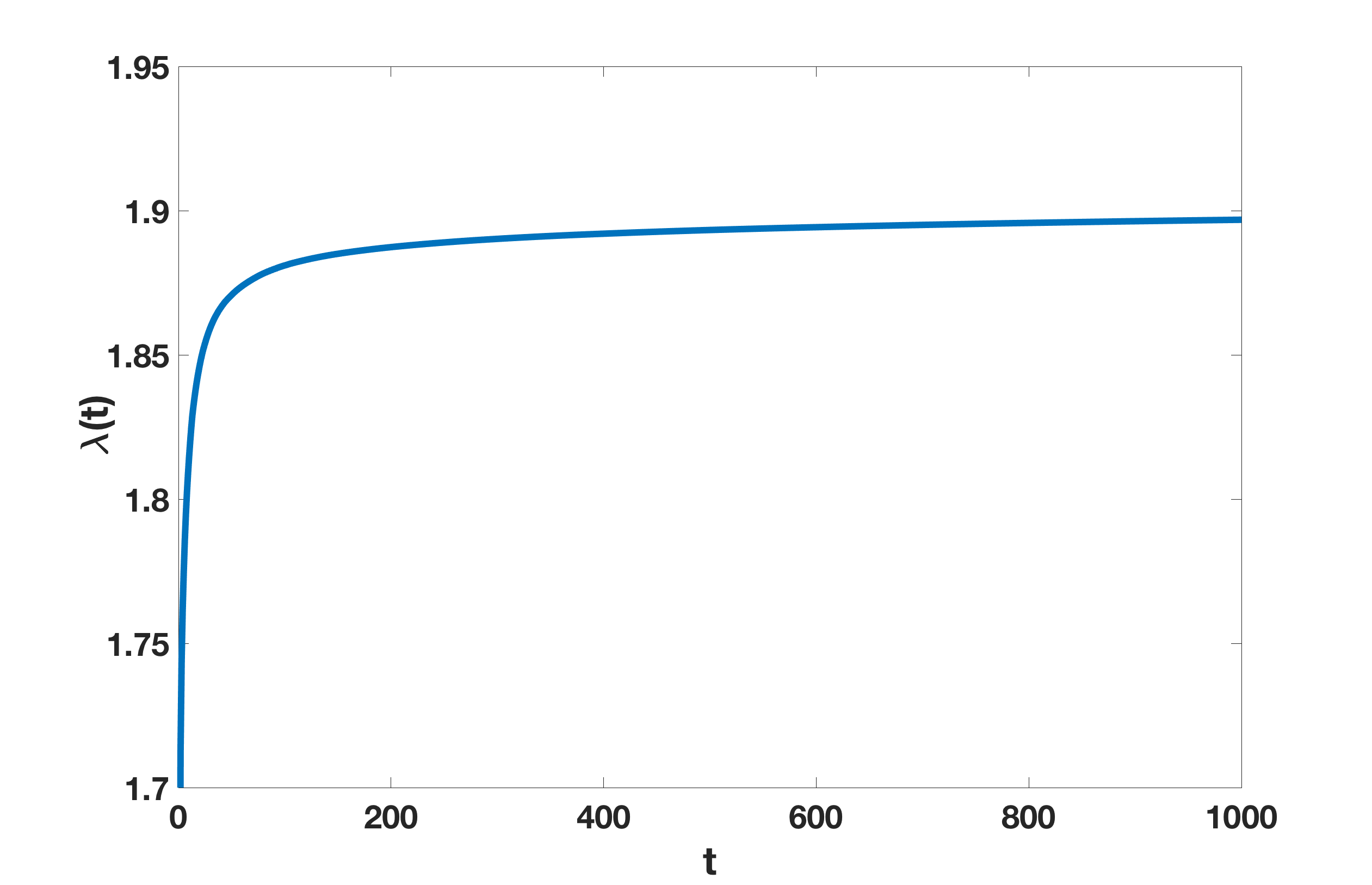}
\end{center}
\end{centering}
\caption{Convergence of $\lambda(t)$   in SEC-L-NK under static attack for $\xi=3$.}
\label{fig:convergence-of-lambda-iteration}
\vspace{-0cm}
\end{figure}

\begin{table}[t!]
\centering
\begin{tabular}{|c |c |c|c|c|c|c|c|c|}
\hline
$\xi$ & 2  & 3  & 4  &  5  &  7  &  10    & 20     \\ \hline
MSE  &  0.0061    &  0.0091    &  0.0076   &  0.0059    &  0.0064    &   0.0066      &   0.0076       \\ \hline
\end{tabular}
\caption{Variation of $MSE(\xi)$ and $\lambda^*(\xi)$ against $\xi$ under static attack.}
\label{table:performance-MSE-and-lambda-against-xi}
\end{table}

\subsubsection{SEC-L for the constrained problem}
Here we consider static attack as before, under the NK situation, but for a different problem instance. We   compute the MSE $\bar{P}_{kalman}$ under a Kalman filter \eqref{eqn:kalman-filter-for-single-sensor} when there is no attack,  set the constraint $\bar{P}=\xi \bar{P}_{kalman}$ (for some factor $\xi > 1$) in \eqref{eqn:constrained-problem-for-estimate-correction}, and then simulate the performance of SEC-L (with the chosen $\bar{P}$).

Figure~\ref{fig:convergence-of-lambda-iteration} illustrates the convergence of $\lambda(t)$ under $\xi=3$. The effect of $\xi$ on the time-average MSE (averaged over $500000$~iterations) of SEC-L-NK (with the two-timescale iteration)   is depicted in Table~\ref{table:performance-MSE-and-lambda-against-xi}. As $\xi$ increases from $2$,  we find that the MSE first increases, then decreases and finally increases again. Hence, picking the optimal value of $\xi$ is crucial in order to minimize MSE under SEC in presence of FDI attack. There is no fixed guideline on how to pick the value of $\xi$ (and consequently $\bar{P}$); this has to be done based on prior knowledge of the system as well as possible attackers. One possible way to choose $\xi$ is to consider the worst possible  linear attack, simulate (off-line) the MSE performance of SEC for various values of $\xi$ in presence of such attacks, find the value of $\xi$ that minimizes the MSE, and then use this value of $\xi$ to SEC used in the real system.


\subsection{Detection under static attack}\label{subsection:detection-performance}
Here we consider that sensors $\{1,2,\cdots,n_0\}$ are under attack (which is not known to the fusion center), and compare the receiver operating characteristic (ROC) curves for DETECT, SAFE and $\chi^2$ detectors.  We run each algorithm for a large number of time slots; the fraction of time slots when the detector is triggered in presence of attack is defined as the detection probability $P_d$, and the fraction of time slots when the detector is triggered in absence of attack is defined as the false alarm probability $P_F$.   In this particular numerical example, we have fixed $J=10$ as the observation window for all detectors. For a range of constraints on the false alarm probability, in each algorithm, we optimize $\eta$   via the LEARN algorithm to meet the false alarm constraint with equality, and use these $\eta$ values to estimate the respective detection probabilities via simulation. It is obvious that the false alarm probability and detection probability both decrease with $\eta$. 

In Figure~\ref{fig:detection-static-attack}, we compare the ROC plots for K and NK cases defined in Section~\ref{subsection:numerical-secure-estimation}. We observe that  DETECT exhibits a much better    detection performance than SAFE and the $\chi^2$ detector (up to $75 \%$ improvement in $P_d$ over SAFE). Also, $\chi^2$ detector outperforms SAFE since SAFE is not an optimal   detection algorithm even with the knowledge of safe sensors.

Note that, DETECT requires us to pre-compute  $\mathbf{\underbar{P}}_{\mathcal{B},\mathcal{B}^c}$ for ${N \choose n_0}$ possible subset pairs; hence, we gain the detection performance improve w.r.t. SAFE (which uses more side information)  at the expense of more computation.

\begin{figure}[t!]
\begin{centering}
\begin{center}
\includegraphics[height=5cm, width=7cm]{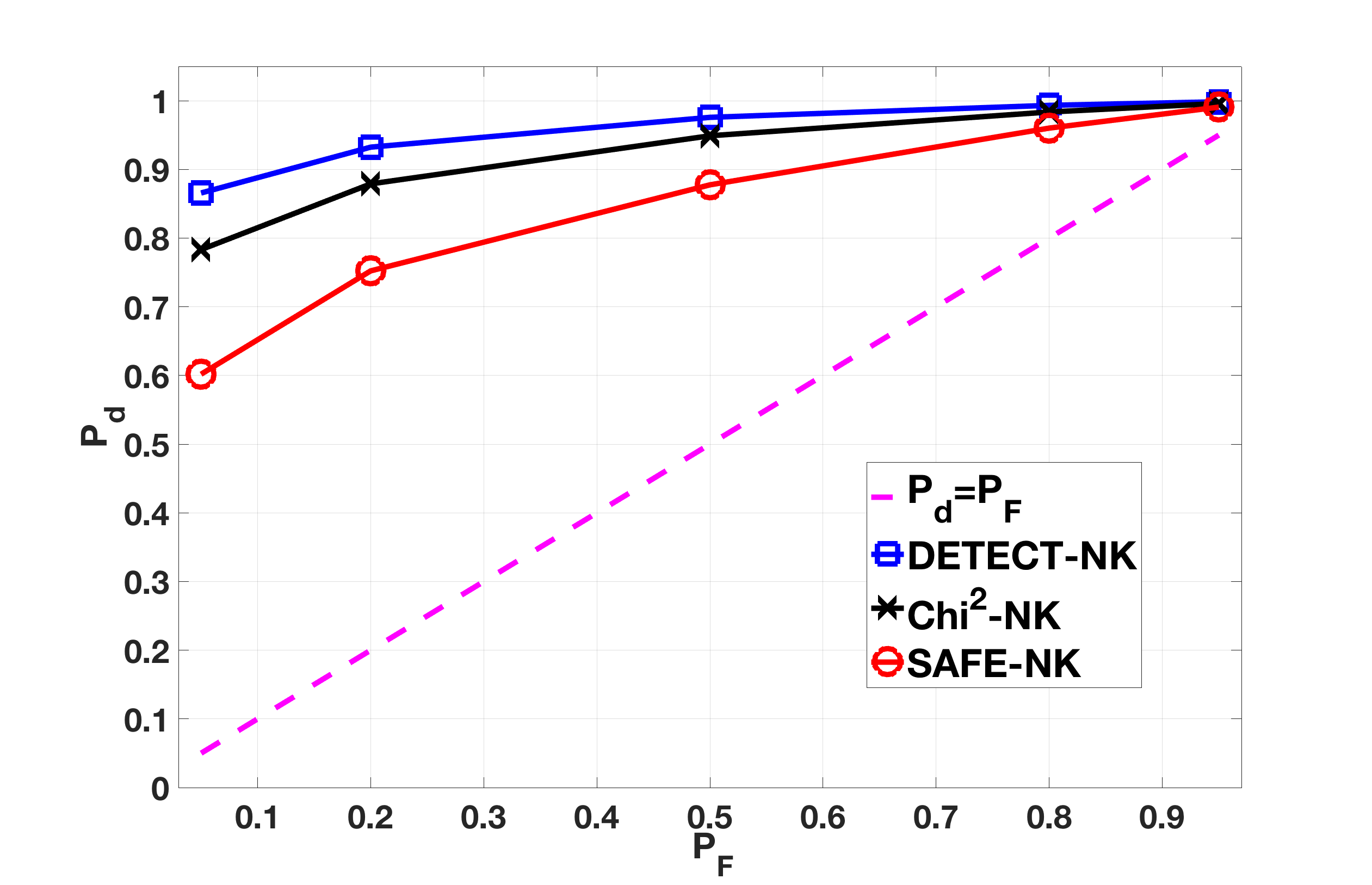}
\includegraphics[height=5cm, width=7cm]{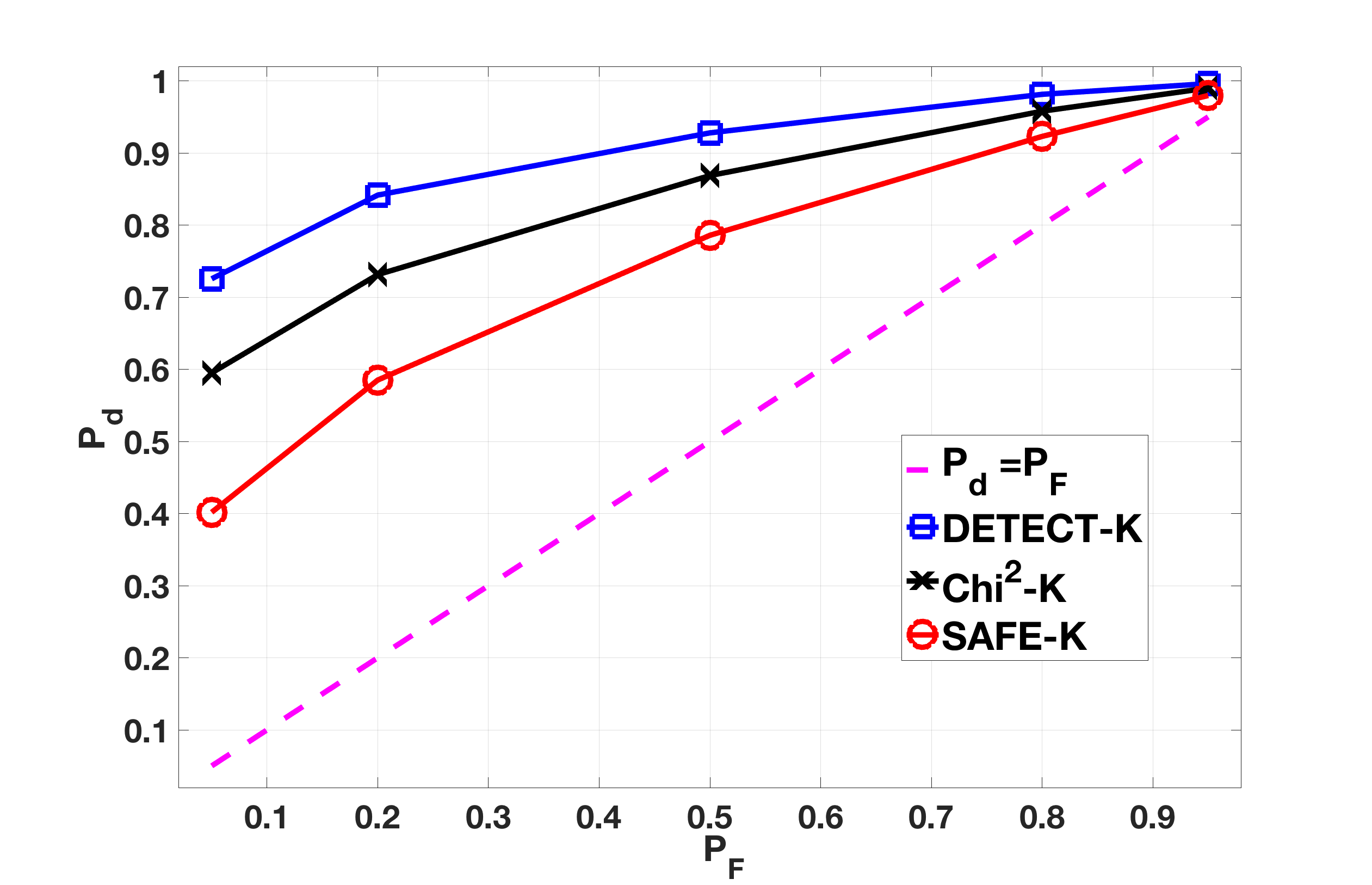}
\end{center}
\end{centering}
\caption{Performance comparison of DETECT against $\chi^2$ and SAFE detectors, under static attack.  $N=5$, $n_0=2$, $k=2$, $q=2$. Two safe sensors are known to the SAFE algorithm. $\eta$ is optimized in all algorithms to meet the false alarm constraint with equality. Top plot: attacker does not know the estimate. Bottom plot: at each time, the attacker knows the previous estimate made at the fusion center.}
\label{fig:detection-static-attack}
\end{figure}

\section{Conclusions}\label{section:conclusion}
In this paper, we    developed a secure estimation algorithm to reduce MSE under FDI attack on a static or time-varying unknown sensor subset, and proved the convergence of the algorithm. Next, we      proposed a detection algorithm for FDI attack on an unknown sensor subset, and also developed another algorithm to optimize the detector. When compared against competing  algorithms, our proposed algorithms demonstrate comparable or even $3$~dB lower MSE and $75\%$ higher attack detection probability subject to a given false alarm constraint,    while using less side information. In future, we would endeavour to expand our work to unknown system dynamics.

\appendices

\section{Proof of Theorem~\ref{theorem:convergence-of-SPSA}}
\label{appendix:proof-of-convergence-of-SPSA}
We consider $b(t)=0$ for all $t=0,1,2,\cdots$, i.e., a constant value of $\lambda$ at each time. 
The proof steps are described in next few subsections. First, we identify the timescales at which various quantities are updated. Next, we prove convergence in the faster timescale and in the slower timescale iterate $\mathbf{K}_t$.
\subsection{Timescale identification}\label{subsection:timescale-identification} Let us recall the evolution of $\mathbf{P}_t$ in \eqref{eqn:evolution-of-error-covariance-for-general-gain-matrix-sequence}, the evolution of $\max_{\mathcal{B} \in 2^{\mathcal{N}}:|\mathcal{B}|=n_0} ||\mathbf{\hat{x}}_{\mathcal{B}}(t)-\mathbf{\hat{x}}_{\mathcal{B}^c}(t)||^2$, and the $\mathbf{K}_t$ update scheme in \eqref{eqn:SPSA-update-equation}. These evolutions  constitute a multi-timescale iteration (see \cite[Chapter~$6$]{borkar08stochastic-approximation-book}).  

Note that, the  $\mathbf{P}_t$ iteration uses a matrix $(\mathbf{I}-\mathbf{K}_t \mathbf{C})$ instead of any diminishing step size sequence, and the spectral radius of $(\mathbf{I}-\mathbf{K}_t \mathbf{C})$ is less than or equal to $(1-\delta)$ for all $t \geq 0$. Hence, the $\mathbf{P}_t$ iteration runs at a faster timescale, and the $\mathbf{K}_t$ iteration (which is a stochastic gradient descent algorithm) runs at a slower timescale since the $\mathbf{K}_t$ iteration uses a diminishing step size sequance $\{a(t)\}_{t \geq 0}$. As a result, the $\mathbf{P}_t$ iteration will view the $\mathbf{K}_t$ iterates as quasi-static, whereas the $\mathbf{K}_t$ iteration will view the $\mathbf{P}_t$ iteration as almost equilibriated. 

 $\mathbf{\hat{x}}_{\mathcal{B}}(t)$ and $\mathbf{\hat{x}}_{\mathcal{B}^c}(t)$  run at a faster timescale compared to the $\mathbf{K}_t$ iteration. Hence,  $\max_{\mathcal{B} \in 2^{\mathcal{N}}:|\mathcal{B}|=n_0} ||\mathbf{\hat{x}}_{\mathcal{B}}(t)-\mathbf{\hat{x}}_{\mathcal{B}^c}(t)||^2$ also runs at a faster timescale.  On the other hand, the $\mathbf{K}_t$ iteration \eqref{eqn:SPSA-update-equation} runs at the slowest timescale.

\subsection{Convergence in the   faster timescale iterates}\label{subsection:convergence-in-faster-timescale} Let us assume that $\mathbf{K}_t=\mathbf{K}$ for all $t$, and that $|\lambda_{max}(\mathbf{I}-\mathbf{K} \mathbf{C})| \leq (1-\delta)<1$. In this case, the $\mathbf{P}_t$ iteration (starting from $\mathbf{P}_0$) becomes:
$$\mathbf{P}_t=(\mathbf{I}-\mathbf{K} \mathbf{C})(\mathbf{A} \mathbf{P}_{t-1}\mathbf{A}'+\mathbf{Q})(\mathbf{I}-\mathbf{K} \mathbf{C})+\mathbf{K}\mathbf{R}\mathbf{K}'$$
These $\mathbf{P}_t$ iterates are bounded, and converge to a positive semidefinite matrix $P(\mathbf{K})$ which is basically the unique fixed point of the above iteration (by Banach's fixed point theorem \cite{rudin76principles-of-mathematical-analysis}, since $|\lambda_{max}(\mathbf{I}-\mathbf{K} \mathbf{C})| \leq 1-\delta$).

Now, it is easy to show that $\{\mbox{Tr}(\mathbf{P}_t)\}_{t \geq 0}$ is a sequence of   continuous, convex  functions of $\mathbf{K}$,  and $\mathbf{K}$ lies in a compact set. Hence,  the limiting function $P(\mathbf{K})$ is     continuous in $\mathbf{K}$.

Under a constant gain matrix $\mathbf{K}$ and a stationary attack, the Markov chains  $\{ \mathbf{\hat{x}}_{\mathcal{B}}(t) \}_{t \geq 0}$ and $\{ \mathbf{\hat{x}}_{\mathcal{B}^c}(t) \}_{t \geq 0}$ converge to their respective steady state distributions exponentially fast in time $t$. Hence,    $\mathbb{E} [ \max_{\mathcal{B} \in 2^{\mathcal{N}}:|\mathcal{B}|=n_0} ||\mathbf{\hat{x}}_{\mathcal{B}}(t)-\mathbf{\hat{x}}_{\mathcal{B}^c}(t)||^2 ] $ converges exponentially fast to  $g(\mathbf{K})$ which is   continuous in $\mathbf{K}$. 

Thus, the expected single stage cost $C(\mathbf{K})= g(\mathbf{K})+\lambda P(\mathbf{K})$   is  continuous in $\mathbf{K}$.

Hence, by an argument similar to the proof of \cite[Chapter~$6$, Lemma~$1$]{borkar08stochastic-approximation-book},  we can write the following relations for  SEC:

\footnotesize
\begin{eqnarray}\label{eqn:convergence-in-faster-timescale}
&& \lim_{t \rightarrow \infty} || \mathbf{P}_t-P(\mathbf{K}_t)||=0 \nonumber\\
&&  \lim_{t \rightarrow \infty} ||\mathbb{E} [ \max_{\mathcal{B} \in 2^{\mathcal{N}}:|\mathcal{B}|=n_0} ||\mathbf{\hat{x}}_{\mathcal{B}}(t)-\mathbf{\hat{x}}_{\mathcal{B}^c}(t)||^2-g(\mathbf{K}_t) ] ||=0 
\end{eqnarray}
\normalsize

\subsection{Convergence of the $\mathbf{K}_t$ iterates}\label{subsection:convergence-in-slower-timescale}
The $\mathbf{K}_t(i,j)$ update   \eqref{eqn:SPSA-update-equation} can be rewritten as: 

\footnotesize
\begin{eqnarray*}
 && \mathbf{\tilde{K}}_{t+1}(i,j) \\
&=& \bigg[\mathbf{K}_t(i,j)-a(t)\frac{c^+(t)-c^-(t)}{2 d(t) \mathbf{\Delta}_t(i,j)}\bigg]_{-l}^l \\
&=& \bigg[\mathbf{K}_t(i,j)-a(t)\times  \frac{\mathbb{E} c^+(t)- \mathbb{E} c^-(t)+M^+(t)-M^-(t)}{2 d(t) \mathbf{\Delta}_t(i,j)}\bigg]_{-l}^l \\
\end{eqnarray*}
\normalsize

\noindent where $M^+(t):=c^+(t)-\mathbb{E} [c^+(t)]$ and $M^-(t):=c^-(t)-\mathbb{E} [c^-(t)]$ are two zero mean Martingale difference sequences. 

Note that, $c^+(t)= \max_{\mathcal{B} \in 2^{\mathcal{N}}:|\mathcal{B}|=n_0}||\mathbf{\hat{x}}_{\mathcal{B}}^+(t)-\mathbf{\hat{x}}_{\mathcal{B}^c}^+(t)||^2 +\lambda \mbox{Tr}(\mathbf{P}_t^+)$. We seek to show that $\mathbb{E}[c^+(t)-c^-(t)]=C(\mathbf{K}_t^+)-C(\mathbf{K}_t^-)+o(d(t))$. 

\begin{lemma} \label{lemma:slower-timescale-first-lemma}
$\mathbb{E} [ \mbox{Tr}(\mathbf{P}_t^+-\mathbf{P}_t^- )]=\mathbb{E} [ \mbox{Tr}(P(\mathbf{K}_t^+)-P(\mathbf{K}_t^-))]+o(d(t))$
\end{lemma}
\begin{proof} Note that,
\footnotesize
\begin{eqnarray*}
\mathbf{P}_t^+ &=& (\mathbf{I}-\mathbf{K}_t^+ \mathbf{C})(\mathbf{A} \mathbf{P}_{t-1}\mathbf{A}'+\mathbf{Q})(\mathbf{I}-\mathbf{K}_t^+ \mathbf{C})'+\mathbf{K}_t^+ \mathbf{R} (\mathbf{K}_t^+)' \\
P(\mathbf{K}_t^+) &=& (\mathbf{I}-\mathbf{K}_t^+ \mathbf{C})(\mathbf{A} P(\mathbf{K}_t^+)\mathbf{A}'+\mathbf{Q})(\mathbf{I}-\mathbf{K}_t^+ \mathbf{C})'+\mathbf{K}_t^+ \mathbf{R} (\mathbf{K}_t^+)'
\end{eqnarray*}
\normalsize

By subtracting these two equations,  we obtain:
\footnotesize
\begin{equation}
\mathbf{P}_t^+-P(\mathbf{K}_t^+)=(\mathbf{I}-\mathbf{K}_t^+ \mathbf{C}) \mathbf{A} (\mathbf{P}_{t-1}-P(\mathbf{K}_t^+))\mathbf{A}' (\mathbf{I}-\mathbf{K}_t^+ \mathbf{C})' \label{eqn:temp1}
\end{equation}
\normalsize

Similarly, we can write:
\footnotesize
\begin{equation}
\mathbf{P}_t^- -P(\mathbf{K}_t^-)=(\mathbf{I}-\mathbf{K}_t^- \mathbf{C}) \mathbf{A} (\mathbf{P}_{t-1}-P(\mathbf{K}_t^-))\mathbf{A}' (\mathbf{I}-\mathbf{K}_t^- \mathbf{C})' \label{eqn:temp2}
\end{equation}
\normalsize

Subtracting \eqref{eqn:temp1} and \eqref{eqn:temp2}, using the Taylor series expansions of $P(\mathbf{K}_t^+)=P(\mathbf{K}_t+d(t)\mathbf{\Delta}_t)$ and $P(\mathbf{K}_t^-)=P(\mathbf{K}_t-d(t)\mathbf{\Delta}_t)$, noting that all iterates are bounded, and simplifying the expressions, we obtain:
\begin{eqnarray*}
&& \mathbf{P}_t^+-\mathbf{P}_t^- -P(\mathbf{K}_t^+)+P(\mathbf{K}_t^-)\\
&=& o(d(t))+O((\mathbf{P}_{t-1}-P(\mathbf{K}_t))d(t))+O(d(t)) \mathbf{\Delta}_t
\end{eqnarray*}
Note that, $\mathbf{P}_{t-1}-P(\mathbf{K}_t)=\mathbf{P}_{t-1}-P(\mathbf{K}_{t-1}+O(a(t-1)))$. Now, since $P(\cdot)$ is a continuous function and since $\lim_{t \rightarrow 0}||\mathbf{P}_{t-1}-P(\mathbf{K}_{t-1})||=0$ (as shown in Appendix~\ref{appendix:proof-of-convergence-of-SPSA}, Section~\ref{subsection:convergence-in-faster-timescale}), and since $\lim_{t \rightarrow 0} a(t-1)=0$, we can claim that $\lim_{t \rightarrow 0} ||\mathbf{P}_{t-1}-P(\mathbf{K}_t)||=0$. Hence, $O((\mathbf{P}_{t-1}-P(\mathbf{K}_t))d(t))=o(d(t))$. Hence, we can write;
\begin{eqnarray*}
 \mathbf{P}_t^+-\mathbf{P}_t^- -P(\mathbf{K}_t^+)+P(\mathbf{K}_t^-)= o(d(t))+O(d(t)) \mathbf{\Delta}_t
\end{eqnarray*}
The $O(d(t))$ term  depends on $\mathbf{K}_t$, which is independent of $\mathbf{\Delta}_t$. Hence, by taking expectation and trace on both sides, and noting that $\mathbb{E} (\mathbf{\Delta}_t)=\mathbf{0}$, the lemma is proved.
\end{proof}

\begin{lemma}\label{lemma:slower-timescale-second-lemma}
Under SEC with $b(t)=0$ for all $t=0,1,2,\cdots$, we have:

\footnotesize
\begin{eqnarray*}
 \mathbb{E} [ \max_{\mathcal{B} \in 2^{\mathcal{N}}: |\mathcal{B}|=n_0} ||\mathbf{\hat{x}}_{\mathcal{B}}^+(t)-\mathbf{\hat{x}}_{\mathcal{B}^c}^+(t)||^2 ] &=&
g(\mathbf{K}_t^+) +o(d(t))\\
 \mathbb{E} [ \max_{\mathcal{B} \in 2^{\mathcal{N}}: |\mathcal{B}|=n_0} ||\mathbf{\hat{x}}_{\mathcal{B}}^-(t)-\mathbf{\hat{x}}_{\mathcal{B}^c}^-(t)||^2 ] &=& g(\mathbf{K}_t^-) +o(d(t)) 
\end{eqnarray*}
\normalsize
\end{lemma}

\begin{proof}
 Let us consider a specific time instant $t=T$, and the corresponding gain matrix $\mathbf{K}_T^+$. 

Let us assume that   a constant gain matrix $\mathbf{K}_{T,\mathcal{B}}^+$ is used for state estimation $\mathbf{\hat{x}}_{\mathcal{B}}^+(t)$ for all time instants  $t =0,1,2,\cdots$. 
\begin{equation*}
\mathbf{\hat{x}}_{\mathcal{B}}^+(t)=(\mathbf{I}-\mathbf{K}_{T,\mathcal{B}}^+ C)A \mathbf{\hat{x}}_{\mathcal{B}}^+(t-1) + \mathbf{K}_{T,\mathcal{B}}^+ \mathbf{y}(t) 
\end{equation*}
Clearly, $\{\mathbf{\hat{x}}_{\mathcal{B}}^+(t)\}_{t \geq 0}$ is a time-homogeneous Markov process in this case, which reaches  its steady state distribution exponentially fast.  Hence, if a gain matrix $\mathbf{K}_{T,\mathcal{B}}^+$ is used to estimate $\{\mathbf{\hat{x}}_{\mathcal{B}}^+(t)\}_{t \geq 1}$ for all subsets $\mathcal{B}$ of size $n_0$, and if a gain matrix $\mathbf{K}_{T,\mathcal{B}^c}^+$ is used to estimate $\{\mathbf{\hat{x}}_{\mathcal{B}^c}^+(t)\}_{t \geq 1}$ for all subsets $\mathcal{B}^c$ of size $N-n_0$, and   these gain matrices are derived from the constant gain matrix $\mathbf{K}_T^+$, then we can write $\mathbb{E} [ \max_{\mathcal{B} \in 2^{\mathcal{N}}:|\mathcal{B}|=n_0} ||\mathbf{\hat{x}}_{\mathcal{B}}^+(T)-\mathbf{\hat{x}}_{\mathcal{B}^c}^+(T)||^2 ] =
g(\mathbf{K}_T^+) +o(d(T))$. This proves the first part of the lemma. The proof of the second part of the lemma follows similarly.
\end{proof}

From Lemma~\ref{lemma:slower-timescale-first-lemma} and Lemma~\ref{lemma:slower-timescale-second-lemma}, we can claim that $\mathbb{E}[c^+(t)-c^-(t)]=C(\mathbf{K}_t^+)-C(\mathbf{K}_t^-)+o(d(t))$.

Hence, by using the results from Lemma~\ref{lemma:slower-timescale-first-lemma}, Lemma~\ref{lemma:slower-timescale-second-lemma}, the iteration \eqref{eqn:SPSA-update-equation} can   be rewritten as:

\footnotesize
\begin{eqnarray*}
&& \mathbf{K}_{t+1}(i,j) = \bigg[\mathbf{K}_t(i,j)-a(t) \times \\
&& \frac{C(\mathbf{K}_t^+)- C(\mathbf{K}_t^-)+M^+(t)-M^-(t)+o(d(t))  }{2 d(t) \mathbf{\Delta}_t(i,j) }\bigg]_{[-l.l]^{q \times m} \cap \mathcal{K}} 
\end{eqnarray*}
\normalsize

Hence, by the discussion in \cite[Section~$2.2$]{borkar08stochastic-approximation-book}, the proof of \cite[Appendix~E, Section~C]{chattopadhyay-etal15measurement-based-impromptu-deployment-arxiv-v1} and the results from \cite{spall92original-SPSA}, one can claim that the effect of $o(d(t))$ is asymptotically negligible, and the iterates of \eqref{eqn:SPSA-update-equation} almost surely converge to the set 
 $\mathcal{K}_{\lambda}:=\{ \mathbf{K}:\nabla_{\mathbf{K}} (C(\mathbf{K}))=0 \}\cap [-l,l]^{q \times m}\cap \mathcal{K}$. This completes the proof of the theorem.

\section{Proof of Theorem~\ref{theorem:convergence-of-SPSA-varying-Lagrange-multiplier}}\label{appendix:proof-of-convergence-of-SPSA-varying-Lagrange-multiplier} 
Note that, the $\lambda(t)$ iteration runs at a slower timescale compared to the $\mathbf{K}_t$ iteration. 

We first note a few things. For any fixed $\lambda \in [0,l]$ used in SEC, we have $\lim_{t \rightarrow \infty}|\mbox{Tr}(\mathbf{P}_t)-P(\mathbf{K}_t)|=0$, as shown in the proof of Theorem~\ref{theorem:convergence-of-SPSA}.  However, $\mathbf{K}_t$ iterates converge to a set $\mathcal{K}_{\lambda}$ for a fixed $\lambda$. Also, $\nabla_{\mathbf{K}} C(\mathbf{K})=\nabla_{\mathbf{K}} g(\mathbf{K})+\lambda \nabla_{\mathbf{K}} P(\mathbf{K})$ is Lipschitz continuous in $\lambda \in [0,l]$, with a uniform Lipschitz constant for all $K \in [-l,l]^{q \times m}$.

Hence, by the discussion of \cite[Section~$2.2$]{borkar08stochastic-approximation-book} and \cite[Chapter~$6$]{borkar08stochastic-approximation-book}, we can say that, under SEC, $D(\mathbf{K}_t ,\mathcal{K}_{\lambda(t)}) \rightarrow 0$ almost surely ($D(\cdot)$ is the distance of a point from a set). Also,  the update $\lambda(t+1)=[\lambda(t)+b(t)(\mbox{Tr}(\mathbf{P}_t)-\bar{P})]_0^l$ asymptotically behaves like $\lambda(t+1)=[\lambda(t)+b(t)(P(\mathbf{K}_t)-\bar{P})]_0^l$. Hence,  by the discussion of \cite[Section~$2.2$]{borkar08stochastic-approximation-book} and \cite[Chapter~$6$]{borkar08stochastic-approximation-book}, we can say that the $\lambda(t)$ update asymptotically evolves according to  the following  stochastic recursive inclusion (see \cite[Chapter~$5$]{borkar08stochastic-approximation-book}) $$\lambda(t+1)\in \{ [\lambda(t)+b(t)(P(\mathbf{K})-\bar{P})]_0^l : \mathbf{K} \in \mathcal{K}_{\lambda(t)}  \}$$

Now we apply the theory of \cite[Chapter~$5$]{borkar08stochastic-approximation-book}  to analysis this stochastic recursive inclusion. Note that, the iteration can as well be represented as:
$$\lambda(t+1)\in \{ [\lambda(t)+b(t)x]_0^l : x \in \mathcal{S}_{\lambda(t)}  \}$$  
where $\mathcal{S}_{\lambda(t)}$ is the closure of the convex hull of the set $\{P(\mathbf{K})-\bar{P}:  \mathbf{K} \in \mathcal{K}_{\lambda(t)} \}$. Clearly, $\mathcal{S}_{\lambda}$ is a closed, convex set for each $\lambda \in [0,l]$; also, $P(\mathbf{K})$ is bounded  since the entries of all possible gain matrices considered in SEC are bounded. Hence,  Conditions~(i) and (ii) of \cite[Section~$5.1$]{borkar08stochastic-approximation-book} are satisfied.  

Now, let us assume that $\lambda(t) \rightarrow \lambda^*$ and $r(t) \rightarrow r^*$, where $r(t) \in \mathcal{S}_{\lambda(t)}$. Note that, $\mathcal{S}_{\lambda(t)}$ is the closure of the convex hull of $\{P(\mathbf{K})-\bar{P}: \mathbf{K} \in [-l,l]^{q \times m}, |\lambda_{max}(\mathbf{I}-\mathbf{K} \mathbf{C})| \leq 1-\delta, \nabla_{\mathbf{K}} C(\mathbf{K},\lambda(t))=0 \}$. Since $\nabla_{\mathbf{K}} C(\mathbf{K},\lambda)$ is a continuous function of $\lambda$, we can easily say that $r^* \in \mathcal{S}_{\lambda^*}$. Hence, Condition~(iii) of \cite[Section~$5.1$]{borkar08stochastic-approximation-book} is satisfied. 

Hence, by \cite[Chapter~$5$, Corollary~$4$]{borkar08stochastic-approximation-book}, we can say that $\lambda(t)$ in SEC almost surely converges to a closed connected internally chain transitive invariant set $\bar{\Lambda}$ of the following differential inclusion:
$$\dot{\lambda}(\tau)\in \mathcal{S}_{\lambda(\tau)}, \tau \in \mathbb{R}^+$$

Since, we have already argued that $||\mathbf{K}_t -\mathcal{K}_{\lambda(t)}|| \rightarrow 0$ almost surely, we can say that $(\mathbf{K}_t,\lambda_t) \rightarrow \{(\mathbf{K},\lambda): \mathbf{K} \in \mathcal{K}_{\lambda}, \lambda \in \bar{\Lambda}\}$. This concludes the proof.

{\small
\bibliographystyle{unsrt}
\bibliography{arpan-techreport}
}

%

\end{document}